\documentclass[a4paper]{jpconf}
\usepackage{amsfonts, amsmath, amssymb, amsthm, amsbsy, amscd}

\newtheorem{lemma}{Lemma}

\theoremstyle{definition}
\newtheorem{example}{Example}

\theoremstyle{remark}

\newcommand{\pinner}{\mathbin{\mathchoice
   {\hbox{\vrule width0.6em depth0pt height0.4pt
   \vrule width0.4pt depth0pt height0.8ex}}
   {\hbox{\vrule width0.6em depth0pt height0.4pt
   \vrule width0.4pt depth0pt height0.8ex}}
   {\hbox{\kern0.14em
   \vrule width0.48em depth0pt height0.4pt
   \vrule width0.4pt depth0pt height0.6ex\kern0.14em}}
   {\hbox{\kern0.1em
   \vrule width0.39em depth0pt height0.4pt
   \vrule width0.4pt depth0pt height0.5ex\kern0.1em}}}}
\newcommand{\inner}{\pinner\,}

\DeclareMathOperator{\img}{im}

\DeclareMathOperator{\Mat}{Mat}
\DeclareMathOperator{\Ann}{Ann}

\newcommand{\tu}{\tilde{u}}

\newcommand{\bE}{\mathbf{E}}
\newcommand{\cE}{\mathcal{E}}
\newcommand{\cEinf}{\mathcal{E}^{\infty}}

\newcommand{\cC}{\mathcal{C}}
\newcommand{\veps}{\varepsilon}
\newcommand{\BBP}{\mathbb{P}}
\newcommand{\BBR}{\mathbb{R}}
\newcommand{\BBN}{\mathbb{N}}
\newcommand{\BBC}{\mathbb{C}}

\newcommand{\bx}{\boldsymbol{x}}
\newcommand{\gm}{\mathfrak{m}}
\newcommand{\dd}{\partial}

\newcommand{\Id}{{\mathrm d}}
\newcommand{\fg}{\mathfrak{g}}

\newcommand{\bw}{{\boldsymbol{w}}}
\newcommand{\bv}{{\boldsymbol{v}}}
\newcommand{\pb}{\dd_{\bw}}

\newcommand{\gl}{\mathfrak{gl}}
\newcommand{\gsl}{\mathfrak{sl}}

\newcommand{\ii}{\boldsymbol{i}}

\theoremstyle{definition}
\newtheorem{definition}{Definition}

\newtheorem{theorem}{Theorem}

\theoremstyle{remark}
\newtheorem{remark}{Remark}

\begin{document}
\pagestyle{plain}

\title
{Gardner's deformation of the Krasil'shchik\/--\/Kersten system}

\author{Arthemy V Kiselev${}^1$ and Andrey O Krutov${}^{2}$}

\address{${}^1$ Johann Bernoulli Institute for Mathematics and Computer Science,
  University of Groningen,
  P.O.Box~407, 9700\,AK Groningen, The Netherlands}
\address{${}^2$ Department of Higher Mathematics, Ivanovo State Power
  University, Rabfa\-kov\-skaya str.~34, Ivanovo, 153003 Russia}

\ead{A.V.Kiselev@rug.nl, 
krutov@math.ispu.ru}

\begin{abstract}
The classical problem of construction of Gardner's deformations 
for infinite\/-\/di\-men\-si\-o\-nal completely integrable systems of evolutionary partial differential equations (PDE) amounts essentially to finding the recurrence relations between the integrals of motion.
Using the correspondence between the zero\/-\/curvature representations 
and Gardner deformations for PDE, 
we construct a Gardner's deformation for the Krasil'shchik\/--\/Kersten system. 
For this, we introduce the new nonlocal variables in such a way that the rules to differentiate them are consistent by virtue of the equations at hand and
second,
the full 
system of Krasil'shchik\/--\/Kers\-ten's equations and the new rules
contains the 
Korteweg\/--\/de Vries 
equation and classical Gard\-ner's deformation for it.
\\[2pt]
\end{abstract}

\section{Introduction}
The search for conservation laws and, in particular, the search for regular methods of construction of
conservation laws~\cite{Miura68,Gardner} are the classical problems in the theory of infinite\/-\/dimensional completely
integrable systems. The existence of infinitely many integrals of motion allows one to detect relevant
bi\/-\/Hamiltonian hierarchies~\cite{Magri}; we refer to~\cite{NewellSolitons} for an almost detective story how
the quest for integrability took place some fifty years ago. Suppose now that ``all'' or ``sufficiently many''
nontrivial conserved currents are known for a PDE~system under study. Then one can use them to balance
solutions at the adjacent faces of a shock wave. Second, each conserved current for a system~$\cE$ with two
independent variables (e.g., the time~$t$ and spatial coordinate~$x$) determines an \emph{Abelian}
nonlocality\footnote{The new variables which we introduce in section~\ref{secGDKK} of this paper over the
  Krasil'shchik\/--\/Kersten system do not belong to the class of Abelian nonlocalities for that model: the
  new structures stem from a zero\/-\/curvature representation, which corresponds to the non\/-\/Abelian picture.}
over~$\cE$; resolving the analytic obstructions, several layers of the new variables introduced in such a way
are often enough for finding a recursion operator or symplectic structures~\cite{JKLstarCov}. Third, the
knowledge of integrals of motion helps one to increase the precision of numerical simulations. Finally, we
recall that advanced analytic methods for solution of Cauchy problems for nonlinear PDEs refer explicitly or
tacitly to various types of conservation in the model at hand. How can the infinite chains of conservation
laws be sought for in a systematic way\,?

Let us specify at once that in this paper we address the problem of finding recurrence relations between the
conserved quantities. In other words, we discard the trivial idea of 
try\/-\/and\/-\/fault search for \emph{isolated}
solutions of the determining equations for generating sections of conservation laws~\cite{BVV}; software for
symbolic calculations~\cite{SsTools,MarvanJets} can be used here. That approach would rely on the standard
techniques in the local geometry of differential equation~\cite{BVV,VinogradovCSpecI,VinogradovCSpecII}; it
often involves various case\/-\/dependent tricks such as the use of 
weight\/-\/homogeneity calculus for scaling\/-\/invariant systems.

Such an algorithm put aside, there remain three strategies to mention. The 
Lax\/-\/pair approach and its
generalisation by Zakharov, Manakov, and Shabat prescribes the calculation of residues for fractional powers
of the spectral operator~\cite{Lax,Faddeev}. Let us note that in the frames of this approach, each conserved
density~$\rho_i$ is obtained without any reference to the already known ones; 
it is then hard to detect any
relation between elements of the infinite sequence with~$i\in\BBN$.

Another option would be to take a ``good'' recursion operator $R\colon \varphi_i \mapsto \varphi_{i+1}$ for
local symmetries of the evolutionary system at hand; the adjoint operator
$R^{\dag}\colon\psi_i\mapsto\psi_{i+1}$ is known to propagate the \emph{cosymmetries} (this is true for the
class of evolutionary models, see~\cite{JKLstarCov}). With a bit of luck, one has that $\psi_i = \bE(\rho_i)$,
meaning that those cosymmetries~$\psi_i$ are true generating sections of nontrivial conserved currents whose
conserved densities are~$\rho_i$ (we denote by~$\bE$ the Euler variational derivative). The densities~$\rho_i$
can be reconstructed by the homotopy formula (see~\cite{Olver} and references therein). It is readily seen
that for bi\/-\/Hamiltonian hierarchies this algorithm can be simplified. Namely, by taking the tower of commuting
flows given by $\varphi_i = J_1(\psi_{i+1}) = J_2(\psi_i)$, one inverts the senior Hamiltonian operator~$J_2$
and then reconstructs the next Hamiltonian functional with density~$\rho_{i+1}$. (The formalism of
Hamiltonian structures for non\/-\/evolutionary systems, which exceeds the frames of this text, was developed by
Krasil'shchik \textit{et al.} in~\cite{JKLstarCov,JKAbel,JKVerbTopical}.) However, we conclude that again, the
second approach yields the hoard of conserved densities~$\rho_i$ but not an explicit relation between them.

Let us recall finally that the renowned seminal paper~\cite{Miura68} led not only to our understanding of the
geometry of zero\/-\/curvature representations and Miura transformations, but also to an important class of
deformation techniques for integrable models~\cite{Fordy83,KuperIrish,TMPh2006}. Gardner's deformations are
the powerful instrument that produces explicit recurrence relations between the integrals of motion (and
moreover, \emph{only} between them, which means that neither the fractional powers of any pseudodifferential
operators are involved nor the fragile homotopy formula is used to invert the Euler variational derivative).

It is the Gardner deformation problem for the Krasil'shchik\/--\/Kersten system which we solve in this paper,
yet it should be emphasised that the solution method which we implement here can be applied to a wide class of
deformation problems for nonlinear models of mathematical physics.

The Krasil'shchik\/--\/Kersten system of two evolution 
equations,\footnote{The choice of sign for the coefficient $u_{12}=-\underline{u}{}_{12}$ reflects the freedom of order, $\theta_1\theta_2\otimes u_{12}=
\theta_2\theta_1\otimes \underline{u}{}_{12}$, in the superfield expansion
$\boldsymbol{u}=1\otimes u_0 + \theta_1\otimes u_1 + \theta_2\otimes u_2 + \theta_1\theta_2\otimes u_{12}$, where $\theta_1\theta_2=-\theta_2\theta_1$ is the product of Grassmann variables.}
\begin{subequations}\label{eqKK}
\begin{align}
\underline{u}{}_{12;t} &={}- \underline{u}{}_{12;xxx} + 6\underline{u}{}_{12}\underline{u}{}_{12;x} - 3u_0u_{0;xxx} - 3u_{0;x}u_{0;xx} + 3\underline{u}{}_{12;x}u_0^2 + 6\underline{u}{}_{12}u_0u_{0;x}, \\
u_{0;t} &={} - u_{0;xxx} + 3u_0^2u_{0;x} + 3\underline{u}{}_{12}u_{0;x} 
+ 3\underline{u}{}_{12;x}u_0.
\end{align}
\end{subequations}
is the bosonic limit of the $N{=}2$ supersymmetric $a{=}1$ Korteweg\/--\/de Vries equation~\cite{MathieuNew}:
to obtain system~\eqref{eqKK} from the super\/-\/equation, one sets equal to zero the two fermionic components~$u_1$,\ $u_2$ of
the superfield~$\boldsymbol{u}$ that contains $N{=}2$ Grassmann variables~$\theta_1$ and~$\theta_2$. The parent $N{=}2$ SKdV equation with $a{=}1$ belongs
to the triplet $a\in\{-2, 1, 4\}$ of completely integrable cases, see~\cite{MathieuNew,MathieuTwo}. The issue
of integrability at $a{=}1$ is difficult, compared to the well\/-\/studied cases $a{=}-2$ and $a{=}4$ (e.g.,
see~\cite{MathieuOpen} and~\cite{PopowiczLax}, also~\cite{N2Hirota
}). The three super\/-\/systems share the
second Hamiltonian operator $J_2^{a{=}-2,1,4}$, whereas the first Hamiltonian structure $J_1^{a{=}1}$ for the case $a{=}1$ is highly non\/-\/local. Kersten and Sorin obtained it in~\cite{KerstenSorin2002} by factorising the
recursion super\/-\/operator $R=J_2^{a{=}-2,1,4}\circ\bigl(J_1^{a=1}\bigr)^{-1}$, thus solving P.~Mathieu's Open problem~5
from~\cite{MathieuOpen} at least formally. (It still remains to verify that the non\/-\/local super\/-\/operator~$J_1^{a{=}1}$ is skew\/-\/adjoint and endows the tower of super\/-\/Hamiltonians with the Poisson bracket; it remains
also to inspect whether the symmetries of the $a{=}1$ SKdV equation produced by such mapping~$R$ remain local.) It is clear however that the bosonic\/-\/limit system~\eqref{eqKK} inherits the bi\/-\/Hamiltonian structure 
$\left( J_1^{a{=}1}{\bigr|}_{\text{Fermi}:=0}, J_2^{a{=}-2,1,4}{\bigr|}_{\text{Fermi}:=0} \right)$ from the parent super\/-\/equation.

To the best of our knowledge, the study of standard geometric structures related to~\eqref{eqKK} was initiated by Krasil'shchik and Kersten
in~\cite{JKKerstenEq} and continued in~\cite{JKLstarCov}, where the proper reduction of the second Hamiltonian
structure for P.~Mathieu's $N{=}2$ SKdV super\/-\/equation was re\/-\/discovered; a recursion operator for symmetries
of~\eqref{eqKK} was obtained via the introduction of suitable nonlocalities,
c.f.~\cite{Sakovich2004} in this context. 
(Let us note that the coefficients of that recursion operator depend on the new nonlocal variables so that the locality of such
operator's output is arguable; yet it could well be that system~\eqref{eqKK} is but a precursor to the larger
model with physical applications.) Around the same time, 
Karasu\/-\/Kalkanl\i\ \textit{et al.}~\cite{KarasuSakovichYurduen2003} approached system~\eqref{eqKK} with the Painlev\'e test, performing
the singularity analysis, and constructed an $\gsl_3(\BBC)$-\/valued zero\/-\/curvature
representation~$\alpha^{\text{KK}}_1$ for~\eqref{eqKK}. We shall use this Lie algebra\/-\/valued one\/-\/form for
solving the Gardner's deformation problem of recursive production of Hamiltonians for the hierarchy of the
Krasil'shchik\/--\/Kersten system. 

Our approach is based on the earlier work~\cite{JMP2012,GradNonRem}, see
also~\cite{Fordy83,KuperIrish,RoelofsThesis} and~\cite{TMPh2006}. 
By understanding that zero\/-\/curvature
representations and Gardner's deformations are structures of the same nature within the nonlocal geometry of
PDE, we reformulate the Gardner deformation problem for Krasil'shchik\/--\/Kersten's system in terms of
construction of parameter\/-\/dependent families of new `nonlocal' variables. We require that these nonlocalities
reproduce the classical Gardner deformation from~\cite{Miura68} under the shrinking of extension~\eqref{eqKK}
for the Korteweg\/--\/de Vries equation back to
\begin{equation}\label{kdv}
u_{12;t} = - u_{12;xxx} - 6 u_{12}u_{12;x} \quad\Longleftrightarrow\quad
\underline{u}{}_{12;t} = - \underline{u}{}_{12;xxx} + 6 \underline{u}{}_{12}\underline{u}{}_{12;x}\quad \text{for $\underline{u}{}_{12}=-u_{12}$.}
\end{equation}
We discover that the nonlocalities which encode the gauge class of 
parameter\/-\/dependent zero\/-\/curvature
representation from~\cite{KarasuSakovichYurduen2003} are a key to solution of the problem: they yield the
recurrence relation between the hierarchy of integrals of motion for the Krasil'shchik\/--\/Kersten system.

\section{Basic concept}
Let us 
recall the 
definitions (see~\cite{BVV,Olver,GDE2012} and~\cite{Marvan2002} for detail);
this material is standard so we now fix some notation and review the concept.

\subsection{The geometry of infinite jet space $J^\infty(\pi)$}
Let $M^n$~be a smooth real $n$-\/dimensional orientable
manifold. Consider a smooth 
vector bundle
$\pi\colon E^{n+m} \to M^n$ with $m$-\/dimensional fibres; 
let us construct the space~$J^{\infty}(\pi)$ of infinite jets of sections for~$\pi$. 
Let $\mu_{\bx_0}^k \Gamma(\pi)$ be the space of (local) sections $s\in\Gamma(\pi)$ such that all partial derivatives of~$s$ up to and including order~$k\geqslant0$ vanish at a point~$\bx_0\in M^n$:
\[
\mu_{\bx_0}^k \Gamma(\pi)=\Bigl\{
s\in\Gamma(\pi)\ \Bigr|\ %
\frac{\partial^{|\sigma|}s}{\partial\bx^\sigma}(\bx_0)=0\qquad
\text{for all $\sigma$ such that } 0\leqslant|\sigma|\leqslant k
\Bigr\};
\]
by convention, the zeroth\/-\/order derivative of any function is the function itself. The rules for
transformation of first\/-{} and higher\/-\/order derivatives under local reparametrisations of the independent variables $\bx=(x^1,\ldots,x^n)$ in a chart $U\subseteq M^n$ imply that the space $\mu_{\bx_0}^k \Gamma(\pi)$~is well defined~--- in a coordinate\/-\/free way.
Consider the quotient space of equivalence classes of (local) sections near a point~$\bx_0$,
\[
J^k_{\bx_0}(\pi) = \Gamma(\pi) / \mu^k_{\bx_0} \Gamma(\pi).
\]
The 
space~$J^k(\pi)$ of $k$th~jets of sections for the vector bundle~$\pi$ is the union
\[
  J^k(\pi) = \bigcup_{\bx_0\in M^n} J^k_{\bx_0}(\pi),
\]
naturally equipped with the smooth manifold structure.
The \emph{infinite jet space}~$J^\infty(\pi)$ is the projective limit
\[
J^{\infty}(\pi) = \mathop{\lim_{\longleftarrow}}_{k\to+\infty} J^k(\pi).
\]
A convenient organisation of local coordinates on~$J^\infty(\pi)$ is as follows: 
let $x^i$~be some coordinate system on a chart in the base~$M^n$ and 
denote by~$u^j$ the fibre coordinates 
in the bundle~$\pi$ so that the variables~$u^j$ play the r\^o\-le of unknowns;
one obtains the collection $u^j_\sigma$ of jet variables along fibres of the vector bundle $J^{\infty}(\pi)\to M^n$ 
(here $|\sigma|\geqslant0$ and~$u^j_\varnothing\equiv u^j$).
In particular, we have that $n=2$, $m=1$, $x^1=x$,
$x^2=t$, $u^1=u_{12}$ for KdV equation~\eqref{kdv} 
and $n=2$, $m=2$, $x^1=x$, $x^2=t$, $u^1=\underline{u}{}_{12}$,
$u^2=u_0$ for Krasil'shchik\/--\/Kersten's system~\eqref{eqKK}.

Define the ring of smooth function on $J^{\infty}(\pi)$ as the inductive limit
\[
C^{\infty}(J^{\infty}(\pi)) =  \left\{ 
  f \colon J^{\infty}(\pi)\to \BBR \mid 
f\in C^\infty(M^n)\text{ or }
\exists k\in\BBN_{\geqslant0} \text{ such that } f\in C^{\infty}(J^k(\pi))
\right\}.
\]
For a function from~$C^\infty(J^\infty(\pi))$ we denote by~$[u]$ its differential dependence on finitely many coordinates along the fibre of the infinite jet bundle over~$M^n$: such differential order can be arbitrarily large but it is always finite.
In this setup, the \emph{total derivatives} $D_{x^i}$~are the commuting 
vector fields
\[
D_{x^i}  
=\frac{\dd}{\dd x^i} + \sum_{j=1}^m\sum_{|\sigma|\geqslant0}
u^j_{\sigma\cup\{i\}}\,\frac{\dd}{\dd u^j_\sigma}
\]
on~$J^\infty(\pi)$. Let us denote the total derivatives also by $\tfrac{\Id}{\Id x^i}$, making no distinction between the two ways of notation.

\label{PageFormalIntegrability}
Let us impose some mild restrictions on the class of~PDEs which we deal with and which are given in local coordinates by using the relations\footnote{%
It is very well known that the \emph{differential order} of a (system of partial) differential equation(s) can depend on a choice of the jet space in which the equation(s) is --\,or are\,-- realised by using the jet fibre coordinates. For example, the \emph{second}\/-\/order equation $\cE'=\{u_{xx}=0\}\subset J^2(\pi'\colon\BBR\times\BBR\to\BBR)$ is obviously equivalent to the \emph{first}\/-\/order system $\cE''=\{u_x=v$, $v_x=0\}\subset J^1(\pi''\colon
\BBR\times\BBR^2\to\BBR)$. Therefore, one can safely think that $K=1$ in~\eqref{EqEqnGeneric} whenever the surjectivity $\cE^\infty\to J^{K-1}(\pi)$ is required to outline the class of ``good'' differential equations (see below).
At the same time, the realisation of Krasil'shchik\/--\/Kersten's system
and Korteweg\/--\/de Vries' equation by using~\eqref{eqKK} and~\eqref{kdv} of differential order three also makes no harm.

Let us emphasize that in this paper, we operate with the \emph{covering} structures over PDEs~$\cE$ by viewing them not just as larger systems~$\tilde{\cE}$ of equations satisfying some properties (see Definition~\ref{DefCovering} on p.~\pageref{DefCovering} below) but as the structures indeed; the coverings are built over the underlying systems~$\cE$ that are given in advance. In other words, the formal integrability and other requirements which we describe on pp.~\pageref{PageFormalIntegrability}--\pageref{PageTangent}
refer to the systems~$\cE$ and their prolongations (specifically, to system~\eqref{eqKK} and equation~\eqref{kdv}). But the key idea of the reasoning that follows is the use of several algebraic realisations for the covering structures over~$\cE$. By this we avoid a necessity to re\/-\/write --\,in terms of a larger system of differential order one\,-- the covering equation~$\tilde{\cE}$ formed by third\/-\/order 
Korteweg\/--\/de Vries' equation~\eqref{kdv} and by rules~\eqref{mcov} to differentiate the nonlocality (here, of order one for~\eqref{mcov1} and order two for~\eqref{mcov2}).
Leaving the details of that particular example to a curious reader, we thank the referee for addressing this issue.}
\begin{equation}\label{EqEqnGeneric}
\cE = \left\{ F^\ell (x^i, u^j, \dots, u^j_{\sigma_\ell}) = 0, \quad \ell =   1,\dots, r<\infty,\ 0<|\sigma_\ell|\leqslant K<\infty \right\}.
\end{equation}
Namely, we study only (systems of) partial differential equations\footnote{We shall 
primarily deal with the \emph{evolutionary} systems of KdV-\/type, equipping them
further with the geometric structures such as the nonlocalities, or \emph{coverings}
(see Definition~\ref{DefCovering} on p.~\pageref{DefCovering}).}
which are \emph{formally integrable} 
(cf.~\cite{BVV,Olver} and~\cite{Goldschmidt} by Goldschmidt);
this class of PDEs~$\cE$ is defined as follows. By definition, put $\cE^{(0)}=\cE$.
For every given PDE system~\eqref{EqEqnGeneric} consider its differential consequences
$\cE^{(k)}=\{D_\tau(F^\ell)=0$ $|$ $|\tau|\leqslant k>0\}$ for all admissible~$\ell$.
Let us assume that at each $k>0$, all these differential\/-\/algebraic relations
determine a smooth submanifold in~$J^{K+k}(\pi)$ such that the projection
$\cE^{(k)}\to\cE^{(k-1)}$ yields a fibre 
bundle. Suppose further that the 
inverse limit $\projlim_{k\to\infty}\cE^{(k)}=\cE^\infty$ is a smooth submanifold 
in~$J^\infty(\pi)$; the object~$\cE^\infty$ is called the 
\emph{infinite prolongation}\footnote{A geometric distinction between
the smooth submanifold $\cE^\infty\subseteq J^\infty(\pi)$ and its description by using the
smooth left\/-\/hand sides in the infinite system $D_\tau(F^\ell)=0$ is that the latter
is always defined yet it can describe the \emph{empty set}.
For instance, consider the overdetermined equation 
$\cE=\{u_{xx}=1$,\ $u_y=x^2\}$ for which $(u_{xx})_y=0\neq 2=(u_y)_{xx}$.
Admitting the slightest abuse of language, we shall make no distinction between 
the geometric objects~$\cE^\infty$ and their algebraic descriptions.}
of the underlying PDE system~$\cE\subseteq J^K(\pi)$. Without loss of generality,
let us assume that the mapping $\cE^\infty\to M^n$ under $\pi_{\infty,-\infty}\colon
J^\infty(\pi)\to M^n$ is onto (otherwise, shrink the base manifold~$M^n$); assume also that
the projection $\cE^\infty\to J^{K-1}(\pi)$ under $\pi_{\infty,K-1}\colon J^\infty(\pi)
\to J^{K-1}(\pi)$ is a surjection as well.\footnote{
For example, the equation $\cE=\{v_x=u$,\ $v_y=u\}$ can be solved only if the compatibility
condition $v_{xy}=v_{yx}$ is satisfied, thus $u_x=u_y$ is the constraint due to which the
projection from $\cE^{(1)}$ down to~$\cE$ is not onto, hence not a vector bundle%
.}
By imposing some geometric conditions on the symbols of such PDEs,
Goldschmidt proves in~\cite[Theorem~9.1]{Goldschmidt} that formally integrable PDE systems
with (locally-) analytic left\/-\/hand sides do possess (local-) analytic solutions
for all Cauchy's data.\footnote{The books~\cite{BVV,Olver} contain an extensive study
of the properties which partial differential equations --\,not only evolutionary\,-- 
must have in order to admit formal solutions and possess infinitesimal symmetries.} 



For \emph{evolutionary} PDE systems
(e.g., for equation~\eqref{kdv} and system~\eqref{eqKK})
the spatial derivatives $u^i_{;\varnothing}\equiv 
u^i$,\ $u^i_{;\bx}$,\ $u^i_{;\bx\bx}$,\ $\ldots$ of all orders together with 
the independent variable(s)~$\bx$ and time~$t$ constitute the collection 
of convenient internal coordinates on the set~$\cE^\infty$, 
which is readily seen to be a 
smooth submanifold in~$J^\infty(\pi)$.
Let us denote by~$\bar{D}_{x^i}$ the restrictions of total derivatives~$D_{x^i}$ to the infinite prolongation~$\cEinf$.
\label{PageTangent}
Thanks to the assumptions which were made in the 
preceding paragraph, these vector fields 
are tangent to the nonsingular submanifold~$\cEinf\subset J^\infty(\pi)$, spanning the Cartan distribution~$\cC\subset T\cE^\infty$ on~it.
At every point $\theta^\infty\in\cEinf$ the tangent 
space~$T_{\theta^\infty}\cEinf$ splits in a direct sum of two subspaces. 
The one which is spanned by the Cartan distribution on~$\cEinf$ is \emph{horizontal} 
and the other is \emph{vertical}, forming the kernel of the differential of the projection
$\cE^\infty\to M^n$; we have that $T_{\theta^\infty} \cEinf = \cC_{\theta^\infty} \oplus
V_{\theta^\infty} \cEinf$. 
We denote by~$\Lambda^{1,0}(\cEinf) = \Ann \cC$ 
and~$\Lambda^{0,1} (\cEinf) = \Ann V\cEinf$ 
the $C^\infty(\cEinf)$-\/modules of contact and horizontal
one\/-\/forms which vanish on~$\cC$ and~$V\cEinf$, respectively.
Denote further by~$\Lambda^r(\cEinf)$ the 
$C^{\infty}(\cEinf)$-\/module of $r$-forms on~$\cEinf$.
There is a natural decomposition $\Lambda^r(\cEinf) = \bigoplus_{q+p = r}
\Lambda^{p,q}(\cEinf)$, where ${\Lambda^{p,q} (\cEinf) = \bigwedge^p \Lambda^{1,0}(\cEinf) \wedge \bigwedge^q \Lambda^{0,1}(\cEinf)}$. This implies that the de Rham differential~$\bar{\Id}$ on~$\cEinf$ is subjected to the decomposition $\bar{\Id} = \bar{\Id}_h + \bar{\Id}_{\cC}$, where $\bar{\Id}_h \colon \Lambda^{p,q}(\cEinf) \to \Lambda^{p,q+1}(\cEinf)$ is the horizontal differential and $\bar{\Id}_{\cC} \colon \Lambda^{p,q}(\cEinf) \to \Lambda^{p+1,q}(\cEinf)$ is the vertical differential.
Let $f\bigl(x^j,[u]\bigr)$ be a function (of finite differential order) on the infinite prolongation~$\cE^\infty$. The horizontal differential $\bar{\Id}_h$ acts on it by the rule $f\mapsto \sum\nolimits_{i=1}^n {\bar{D}}_{x^i}(f)\,\Id x^i$. This formula's extension to the spaces $\Lambda^{0,q}(\cE^\infty)$ of horizontal $q$-\/forms is immediate: for any $\eta=f\bigl(x^j,[u]\bigr)\,\Id x^{i_1}\wedge\ldots\wedge\Id x^{i_q}$ we have that $\bar{\Id}_h(\eta)=\sum\nolimits_{i=1}^n {\bar{D}}_{x^i}(f)\,\Id x^i\wedge\Id x^{i_1}\wedge\ldots\wedge\Id x^{i_q}\in\Lambda^{0,q+1}(\cE^\infty)$.\label{DefDH}
For $p>0$, the action of horizontal differential~$\bar{\Id}_h$ on the space $\Lambda^{p,q}(\cE^\infty)$ of differential forms containing $p$~Cartan's differentials is highly nontrivial%
\footnote{
The horizontal differential $\bar{\Id}_h$ acts on the spaces $\Lambda^{p,q}(\cE^\infty)$ of differential forms via the graded Leibniz rule; its application to Cartan's forms $\bar{\Id}_{\cC}(u^j_\sigma)$ is deduced from the identity~$\bar{\Id}^2=0$ for the de Rham differential $\bar{\Id}=\bar{\Id}_h+\bar{\Id}_{\cC}$ on~$\cE^\infty$. Specifically, from $\bar{\Id}_h^2=\bar{\Id}_h\circ\bar{\Id}_{\cC}+\bar{\Id}_{\cC}\circ\bar{\Id}_h=\bar{\Id}_{\cC}^2=0$ one infers that 
$\bar{\Id}_h\circ\bar{\Id}_{\cC}=-\bar{\Id}_{\cC}\circ\bar{\Id}_h$, thus reducing the action of~$\bar{\Id}_h$ to the case when it has already been defined. In brief, the formula
$\bar{\Id}_h = \sum\nolimits_i \Id x^i \wedge {\bar{D}}_{x^i}$
means that the vector fields $\bar{D}_{x^i}$ proceed by the Leibniz rule over the argument's wedge factors, acting on each factor --\,pushed leftmost\,-- via the Lie derivative.} 
(see~\cite{VinogradovCSpecI, VinogradovCSpecII}, also~\cite{JKVerbTopical}).
However, in this paper
we deal with the forms that contain no Cartan's differentials.
By definition, we put $\bar{\Lambda}(\cE^\infty)=\bigoplus_{q\geqslant0}\Lambda^{0,q}
(\cE^\infty)$ and we denote by $\overline{H}^n(\cE^\infty)$ the senior $\bar{\Id}_h$-\/cohomology
group (also called senior \emph{horizontal cohomology}). 


A conserved current $\eta$ for the system $\cE$ is the continuity equation
\[
  \sum_{i=1}^{n} \bar{D}_{x^i} (\eta_i) \doteq 0 \ \text{on}\ \cEinf,
\]
where the symbol~$\doteq$ denotes the equality by virtue of the system~$\cE$ 
and its differential con\-se\-quen\-ces.
The quantities~$\eta_i(x^j, [u^k])$ 
convene to the $\bar{\Id}_h$-\/closed horizontal $(n-1)$-\/form
\[
\eta = \sum_{i=1}^{n}(-1)^{i+1} \eta_i \cdot \Id x^1 \wedge \ldots \wedge \widehat{\Id x^i} \wedge \ldots
\wedge \Id x^n \in \bar{\Lambda}^{n-1}(\pi),
\]
in which the wedge factors~$\widehat{\Id x^i}$ are omitted;
the conservation of the current~$\eta$ is the equality
$\bar{\Id}_h
(\eta) \doteq 0$ on~$\cEinf$.
The coefficient $\eta^n$ is called the \emph{conserved density} and coefficients $\eta^1$,\ $\dots$,\ $\eta^{n-1}$ are the \emph{flux} components.
By definition, a current~$\eta$ is \emph{trivial} if it is $\bar{\Id}_h$-\/exact: $\eta=\bar{\Id}_h(\xi)$ for some $(n-2)$-\/form~$\xi$ on~$\cE^\infty$ (here $n\geqslant2$).
A \emph{conservation law} $\int\eta\in\overline{H}^{n-1}(\cE^\infty
)$ for an equation $\cE$ is the equivalence class of conserved
currents $\eta$ taken modulo globally defined\footnote{%
Let us exclude --\,from the future consideration\,-- the \emph{topological} conservation laws (which arise from the geometry of the bundle~$\pi$ or from the topology $H^{n-1}(\cE^\infty)\neq0$ of the PDE system at hand, cf.~\cite{BVV} or~\cite{GDE2012}).
This is legitimate for the geometry of Korteweg\/--\/de Vries' equation and Krasil'shchik\/--\/Kersten's system under study; we have that the vector bundles $\pi\colon\BBR^n\times\BBR^m\to\BBR^n\mathrel{{=}{:}}M^n$ are topologically trivial in the both cases.

The referee recalls that the requirement for a trivial current~$\bar{\Id}_h(\xi)$ to be globally defined does not exclude a possibility for existence of the topological conservation laws. For example, let the $n$-\/dimensional base manifold~$M^n$ be such that its $(n-1)$\/th de Rham cohomology group is nonzero; now pick any closed differential $(n-1)$-\/form $\omega$ on~$M^n$ such that the de Rham cohomology class of~$\omega$ is nonzero. Next, construct the trivial bundle $\pi\colon M^n\times\BBR\to M^n$ with a coordinate~$u$ in the fibre, and postulate the PDE $\cE=\{u=0\}$ so that the section $u=0$ is its only solution. Then $\cE^\infty\cong M^n$, and one can regard $\omega$ as a horizontal $(n-1)$-\/form on~$\cE^\infty$. Note that, specifically to this example, the horizontal differential~$\bar{\Id}_h$ on~$\cE^\infty\cong M^n$ is equal to the usual de Rham differential, so that~$\bar{\Id}_h(\omega)=0$. Since the de Rham cohomology class of~$\omega$ is nonzero, there does not exist a globally defined form~$\xi$ on~$M^n\cong\cE^\infty$ satisfying~$\bar{\Id}_h(\xi)=\omega$. (Let us note that the globally defined topological conserved current~$\omega$ on~$\cE^\infty$ is \emph{not trivial} by the definition of cohomology.) The referee concludes that the conservation law $\int\omega\in\overline{H}^{n-1}(\cE^\infty)$ is topological because it depends only on the topology of~$\cE^\infty\cong M^n$. However, all of this is irrelevant to the problem that we deal with in this paper.%
}
exact forms $\bar{\Id}_h \xi \in \int 0$. In other words, two 
conserved currents $\eta_1$ and $\eta_2$ are equivalent if they differ by 
a trivial current: 
$\eta_1 - \eta_2 = \bar{\Id}_h \xi$. 
We denote by $\overline{H}^{n-1}(\cE^\infty)$ the $(n-1)$\/th horizontal cohomology group for~$\cE^\infty$, that is, the set of equivalence classes of conserved currents which is equipped with the structure of Abelian group.



\subsection{Gardner's deformations}

\begin{definition}[\cite{Miura68,KuperIrish,TMPh2006}]\label{DefGardner}
Let $\cE = \left\{ u_t = f(x, [u]) \right\}$
be a system of evolution equations (in particular, a completely integrable system).
Suppose $\cE(\veps) = \{ \tu_t = f_\veps(x, [\tu], \veps)  \mid f_\veps \in \img \tfrac{\Id}{\Id x} \}$
is a deformation of $\cE$ such that at each point $\veps \in
\mathcal{I}$ of an interval $\mathcal{I}\subseteq  \BBR$ 
there is the \emph{Miura contraction}
$\mathfrak{m}_\veps = \{ u = u([\tu], \veps)\} \colon
\cE(\veps) \to \cE$. 
Then the pair $(\cE(\veps), \mathfrak{m}_\veps)$
is the (\emph{classical}) \emph{Gardner deformation} for the system~$\cE$.
\end{definition}


Under the assumption that $\cE(\veps)$ be in the form
of a 
conserved current, the Taylor
coefficients $\tilde{u}^{(k)}$ of the formal power series
$\tilde{u}=\mathop{\sum_{k=0}^{+\infty}}\tilde{u}^{(k)}\cdot\veps^k$ are
termwise conserved on~$\cE(\veps)$ and hence on~$\cE$. Therefore,
the contraction~$\gm_\veps$ yields the recurrence relations, ordered
by the powers of~$\veps$, between these densities~$\tilde{u}^{(k)}$,
while the equality $\cE(0)=\cE$ specifies the initial condition for those relations.

\begin{example}[\textup{\textmd{\cite{Gardner}}}]\label{ExKdVe}
The contraction
\begin{subequations}\label{DefKdV}
\begin{align}
\gm_\veps&=\smash{\bigl\{u_{12}=\tu_{12}
   \pm\veps\tu_{12;x}-\veps^2\tu_{12}^2\bigr\}}\label{KdVeKdV}\\
\intertext{maps solutions $\tu_{12}(x,t;\veps)$ of
the extended equation}
\cE(\veps)=\bigl\{
  \tu_{12;t}&+
\bigl(\tu_{12;xx}+3\tu_{12}^2
   -2\veps^2\cdot\tu_{12}^3\bigr)_x=0\bigr\},\label{KdVe}
\end{align}
\end{subequations}
to solutions $u_{12}(x,t)$ of the Korteweg\/--\/de Vries equation
\begin{equation}
\cE=\bigl\{ u_{12;t} = - u_{12;xxx} - 6 u_{12}u_{12;x} \bigr\}.\tag{\ref{kdv}}
\end{equation}
Plugging the 
series $\tu_{12}=\sum_{k=0}^{+\infty}u_{12}^{(k)}\cdot\veps^k$
into 
expression~\eqref{KdVeKdV} 
for~$\tu_{12}$, we obtain the chain of equations ordered by 
powers of~$\veps$,
\[
u_{12}=\sum_{k=0}^{+\infty}\tu_{12}^{(k)}\cdot\veps^k 
  \pm \tu_{12;x}^{(k)}\cdot\veps^{k+1}
   -\sum_{\substack{i+j=k\\ i,j\geq0}} \tu_{12}^{(i)}\tu_{12}^{(j)}\cdot
    \veps^{k+2}.
\]
Let us fix the plus sign in~\eqref{KdVeKdV} by reversing $\veps\to-\veps$ if necessary.
Equating the co\-ef\-fi\-ci\-ents
of
~$\veps^k$, we obtain the relations
\[
u=\tu_{12}^{(0)},\qquad
0=\tu_{12}^{(1)}+\tu_{12;x}^{(0)},\qquad
0=\tu_{12}^{(k)}+\tu_{12;x}^{(k-1)}-\sum_{\substack{i+j=k-2\\ i,j\geq0}}
   \tu_{12}^{(i)}\tu_{12}^{(j)},\quad k\geq2.
\]
Hence, from the initial condition $\tu_{12}^{(0)}= u_{12}$
we recursively generate the densities
\begin{align*}
\tu_{12}^{(1)} &= - u_{12;x},\qquad
\tu_{12}^{(2)} = u_{12;xx} - u_{12}^2,\qquad 
\tu_{12}^{(3)} = - u_{12;xxx} + 4u_{12;x}u_{12},\\
\tu_{12}^{(4)} &= u_{12;4x} - 6u_{12;xx}u_{12} - 5u_{12;x}^2 + 2u_{12}^3,\\
\tu_{12}^{(5)} &= - u_{12;5x} + 8u_{12;xxx}u_{12} + 18u_{12;xx}u_{12;x} - 16u_{12;x}u_{12}^2,\\
\tu_{12}^{(6)} &= u_{12;6x} - 10u_{12;4x}u_{12} - 28u_{12;xxx}u_{12;x} - 19u_{12;xx}^2 + 30u_{12;xx}u_{12}^2 + 50u_{12;x}^2u_{12} - 5u_{12}^4,\\
\tu_{12}^{(7)} &= - u_{12;7x} + 12u_{12;5x}u_{12} + 40u_{12;4x}u_{12;x} + 68u_{12;xxx}u_{12;xx} - 48u_{12;xxx}u_{12}^2\\
 {}&{}\qquad{} - 216u_{12;xx}u_{12;x}u_{12} - 60u_{12;x}^3 + 64u_{12;x}u_{12}^3,\quad \text{etc}.
\end{align*}
The conservation $\tu_{12;t}=\tfrac{\Id}{\Id x}\bigl(\cdot\bigr)$ implies
that each coefficient~$\tu_{12}^{(k)}$ is conserved on~\eqref{kdv},
and one proves easily that the densities~$\tu_{12}^{(2k)}$ with even 
indexes~$2k\in 2\BBN_{\geqslant0}$ determine the hierarchy of nontrivial conservation 
laws for the Korteweg\/--\/de Vries equation (\cite{Miura68,Gardner} vs~\cite{Magri}).
\end{example}

\subsection{Zero\/-\/curvature representations}
Let $G$~be a finite\/-\/dimensional matrix complex Lie group and $\fg$~be its Lie algebra.
Consider the tensor product $\fg\mathbin{{\otimes}_{\mathbb{R}}}\bar{\Lambda}(\cEinf)$
of~$\fg$ 
with 
the exterior algebra
$\bar{\Lambda}(\cEinf)=\bigoplus_i \Lambda^{0,i}(\cEinf)$.
The product is endowed with the bracket 
\[ 
 [ A\otimes\mu, B\otimes\nu ] =  [ A,B ]\otimes\mu\wedge\nu
\] 
for $A,B\in \fg$ and $\mu,\nu\in\bar{\Lambda}(\cEinf)$.
Define the operator $\bar{\Id}_h$ that acts on elements of $\fg\otimes\bar{\Lambda}(\cEinf)$ by the rule 
\[
\bar{\Id}_h(A\otimes\mu) = A\otimes\bar{\Id}_h\mu,
\]
where the horizontal differential~$\bar{\Id}_h$ in the right\/-\/hand side is~
already defined (see p.~\pageref{DefDH}).
Elements of $\fg\otimes C^\infty(\cEinf)$ are called
\emph{$\fg$-\/matrices}~\cite{Marvan2002}.

\begin{definition}[\cite{Marvan2002,Marvan2010}]\label{DefZCR}
A horizontal $1$-\/form $\alpha\in\fg\otimes\bar{\Lambda}^1(\cEinf)$ is
called a $\fg$-\/valued \emph{zero\/-\/curvature representation} (ZCR) for
the equation~$\cE$ if the Maurer\/--\/Cartan condition,
\begin{equation}\label{zcrf}
\bar{\Id}_h\alpha \doteq \tfrac12 [ \alpha,\alpha ],
\end{equation}
holds by virtue of~$\cE$ and its differential consequences.
\end{definition}

\begin{example}\label{exLaxZCR}
The $\gsl_2(\BBC)$-\/valued zero\/-\/curvature representation
$\alpha_1^{\text{KdV}} = A\,\Id x + B\,\Id t$
for KdV equation~\eqref{kdv},
\begin{equation}
\label{zcrBVV}
\alpha_1^{\text{KdV}} = \begin{pmatrix}
0 & \lambda - u_{12}\\
1 & 0
\end{pmatrix}\Id x
+
 \begin{pmatrix}
- u_{12;x} & - 4\lambda^2 + 2\lambda u_{12}  + 2u_{12}^2 +
  u_{12;xx}  \\
 - 4\lambda - 2u_{12}  &  u_{12;x}
\end{pmatrix} \Id t,
\end{equation}
is known from 
the 
paper~\cite{ZSh}.
\end{example}

Recall that $\fg$~is the Lie algebra of a given Lie group~$G$, see above.
Elements of $C^{\infty}(\cEinf, G)$, i.e., 
$G$-\/valued functions on~$\cEinf$, are called \emph{$G$-\/matrices}.
Let $\alpha$ and $\alpha^\prime$ be $\fg$-\/valued zero\/-\/curvature representations, then $\alpha$ and $\alpha^\prime$ are
called \emph{gauge\/-\/equivalent} if there exists a $G$-\/matrix 
$S\in C^{\infty}(\cEinf, G)$ such that 
\begin{equation}\label{gaugetrans}
\alpha^{\prime} = \bar{\Id}_hS\cdot S^{-1} + S\cdot\alpha\cdot S^{-1}
\mathrel{{=}{:}} \alpha^S .
\end{equation}
Not only that the notions of Gardner's deformations and zero\/-\/curvature representations are intimately related but moreover, the idea of gauge equivalence allows us to revise, simplify, and solve the deformation problem 
for~\eqref{eqKK}.

\subsection{Differential coverings}
We now recall the definition of differential covering over (the infinite prolongation of)
a given PDE system~$\cE$. This notion brings together the procedure of Gardner's 
deformations and the construction of zero\/-\/curvature representations, allowing
--\,in principle\,-- to interpret both of them in terms of new, larger set of
differential equations that contains~$\cE$ as sub\/-\/system. (By construction, the new
differential equations will not necessarily be evolutionary even if the system~$\cE$ is.)
Let us remark however that in this text, the coverings over evolution equations are
viewed and operated with as differential\/-\/geometric structures over the underlying
systems~(\ref{eqKK}-\ref{kdv})
of KdV-\/type.\footnote{This approach is highlighted by using the notation:
nominally realised as differential equations, the coverings' total spaces~$\tilde{\cE}$ 
fibre over the infinite prolongations~$\cE^\infty$ of evolutionary systems.}
More specifically, the covering structures will be described by using matrix Lie groups 
and algebras and by using Lie subalgebras in the algebras of vector fields on some fibre
bundles over the manifolds~$\cE^\infty$.

\begin{definition}[\cite{BVV,JKAMNonlocalTrends}]\label{DefCovering}
A \emph{covering} (or \emph{differential covering}) over 
a formally integrable equation~$\cE$ is another (usually,
larger) system of partial differential equations $\tilde{\cE}$
endowed with the $n$-\/dimensional Cartan distribution $\tilde{\cC}$
and such that there is a mapping $\tau\colon\tilde{\cE}\to\cEinf$ for which at each point
$\theta$ of the manifold~$\tilde{\cE}$, the tangent map $\tau_{*,\theta}$ is an isomorphism of the plane
$\tilde{\cC}_{\theta}$ to the Cartan plane $\cC_{\tau(\theta)}$ at the point~$\tau(\theta)$ in $\cEinf$.
\end{definition}

The construction of a covering over~$\cE$ means the introduction of new variables in such a way that the compatibility of their mixed derivatives is valid by virtue of the underlying~$\cEinf$. 
In practice (see~\cite{GDE2012}), it is the rules to
differentiate the new variable(s) which are specified in a consistent
way; this implies that those new variables acquire the nature of
nonlocalities if their derivatives are local but the variables
themselves are not (e.g., consider the potential $\mathfrak{v}=\int
u_{12}\,\Id x$ satisfying $\mathfrak{v}_x = u_{12}$ 
and $\mathfrak{v}_t = -u_{12;xx} -
3u_{12}^2$ for the KdV equation $u_{12;t} + u_{12;xxx} + 6u_{12}u_{12;x} = 0$).
Whenever the covering is indeed realised as the fibre bundle
$\tau\colon\tilde{\cE}\to\cE$, the forgetful map~$\tau$ discards the
nonlocalities.

In these terms, zero\/-\/curvature representations and Gardner's
deformations are coverings of special kinds\footnote{The link 
between zero\/-\/curvature representations and recursion operators is discussed in the papers~\cite{Sakovich2004,BaranMarvan2008}.}
(see Examples~\ref{exCovAndZcr} and~\ref{exGardnerCov} below).
Each zero\/-\/curvature
representation with coefficients belonging to a matrix Lie algebra 
determines a (linear) covering, whereas each covering with 
fibre~$W$ can be regarded as a zero\/-\/curvature representation 
whose 
coefficients 
take values in the Lie algebra of vector fields on~$W$.
Indeed, let $x^1$,\ $\ldots$,\ $x^n$ be the independent variables in a given PDE and ${\bar{D}}_{x^i}$~be the
corresponding total derivative operators. 
Then the zero\/-\/curvature representations and coverings are described
by the same equation~\eqref{zcrf},
\[
[\bar{D}_{x^i}+A_i, \bar{D}_{x^j}+A_j]\doteq0, \qquad 
1\leqslant i<j\leqslant n. 
\]
In the case of zero\/-\/curvature representations, the coefficients $A_i$ and $A_j$ are functions on~$\cE^\infty$ taking values in a Lie algebra.  
In the case of coverings, the objects $A_i$ and $A_j$ are vertical vector fields on the
covering manifold.
This correspondence between zero\/-\/curvature representations and coverings very often allows one to transfer
results on ZCRs to results on coverings and \textit{vice versa},
see~\cite{GradNonRem,KKIgonin2003} and~\cite{
PopovychBoykoNesterenko2003,Shchepochkina2006} for detail.

The use of geometric similarity of the two notions allows us to construct 
new Gardner's deformations from known zero\/-\/curvature representations 
that take values in finite\/-\/dimensional complex Lie algebras.%
\footnote{Whenever the vector field realisation of a covering structure over~$\cE^\infty$
is given \emph{a priori}, the problem of reconstruction and recognition of a Lie
algebra that could determine that covering is nontrivial; e.g., take
the vector fields that encode a Gardner deformation for~$\cE$ in terms of a covering
over its infinite prolongation. Suppose for definition that the coefficients of these 
vector fields are polynomial in the nonlocal variables. Then, as soon as one starts taking
the fields' iterated commutators, either they close to a manifestly finite\/-\/dimensional
Lie subalgebra in the Lie algebra of 
vector fields on the covering's fibre --- or the degrees of such polynomials grow 
infinitely. In that situation, the polynomials of different degrees determine linearly
independent elements within a basis of the 
generated Lie subalgebra in the Lie
algebra of all vector fields on the fibres. 
For example, such is the case of Gardner's deformation 
for the Kaup\/--\/Boussinesq equation~\cite{HKKW}.
However, there still remain two possible options: either the Lie algebra realised by using all the commutator\/-\/generated vector fields is truly 
in\-fi\-ni\-te\/-\/di\-men\-si\-o\-nal
or there is a \emph{finite}-\/dimensional Lie algebra such that the vector fields on 
the covering's fibres provide its infinite\/-\/dimensional representation.}
%

\begin{example}[Zero\/-\/curvature representations as 
coverings]\label{exCovAndZcr}
Let $\fg\mathrel{{:}{=}}\mathfrak{sl}_2(\BBC)$ as in Example~\ref{exLaxZCR} 
and introduce the standard basis $e,h,f$ in~$\fg$ so that
\[
 [e,h] = -2e,\quad [e,f] = h,\quad [f,h] = 2f.
\]
Let us consider the matrix representation
\[
\rho:\ \mathfrak{sl}_2(\BBC)\to \{ A \in \Mat(2,2) | \tr A =0\}
\]
of~$\fg$ and, simultaneously, 
its representation $\varrho$ in the space of vector fields with
polynomial coefficients on the complex line with the coordinate~$w$:
\[
\begin{array}{rclrclrcl}
\rho(e)& = &\begin{pmatrix} 0 & 1 \\ 0 & 0 \end{pmatrix}, &
\rho(h)& = &\begin{pmatrix} 1 & \phantom{+}0 \\ 0 & -1 \end{pmatrix},&
\rho(f)& = &\begin{pmatrix} 0 & 0 \\ 1 & 0 \end{pmatrix},\\
\varrho(e)& = &1\cdot\partial/\partial w, \quad &
\varrho(h)& =&-2w\cdot\partial/\partial w, \quad &
\varrho(f) &=&-w^2\cdot\partial/\partial w.
\end{array}
\]
Let us decompose the matrices $A_i\in C^\infty(\cEinf)\otimes \fg$, occurring in the zero\/-\/curvature representation $\alpha = \sum_i A_i\,\Id x^i$,
with respect to the basis in the space~$\rho(\fg)$:
\[
A_i  = a_e^{(i)} \otimes \rho(e) + a_h^{(i)} \otimes \rho(h) + a_f^{(i)} \otimes \rho(f)\qquad \text{for
  $a^{(i)}_j \in C^{\infty}(\cEinf)$.}
\]
To construct the covering $\tilde{\cE}$ over $\cEinf$ with a new
fibre variable~$w$ 
(the `nonlocality'), we switch from the representation~$\rho$ to~$\varrho$. 
We thus obtain the map 
\begin{equation}\label{atensvf}
  A_i \mapsto V_{A_i}
\end{equation}
that takes the $\fg$-\/matrices~$A_i$ to the vector fields
\[
V_{A_i} =  a_e^{(i)} \otimes \varrho(e) + a_h^{(i)} \otimes \varrho(h) + a_f^{(i)} \otimes \varrho(f);
\]
the prolongations of total derivatives 
$D_{x^i}$ to~ $\tilde{\cE}$ are now defined by the formula
\begin{equation}\label{tilded}
\tilde{D}_{x^i} = D_{x^i} - V_{A_i}.
\end{equation}
The extended derivatives act on the nonlocal variable $w$ as follows,
\[
\tilde{D}_{x^i} w = \Id w \inner ( - V_{A_i}).
\]
We shall use this approach in the construction of the covering 
in Example~\ref{Ex1DoverKdV}, see below.
\end{example}

\begin{remark}
The commutativity ${\bigl[\tilde{D}_{x^i}, \tilde{D}_{x^j}\bigr] = 0}$
of 
prolonged total derivatives for all $i\neq j$
is equivalent to Maurer\/--\/Cartan's equation~\eqref{zcrf}.
Indeed, we have that
\begin{multline*}
0 = [\tilde{D}_{x^i}, \tilde{D}_{x^j}] = [D_{x^i} - V_{A_i}, D_{x^j} -
V_{A_j}] = [D_{x^i}, D_{x^j}]  - [D_{x^i}, V_{A_j}] - [V_{A_j},
D_{x^j}] + [V_{A_i}, V_{A_j}]= \\
{}= - V_{D_{x^i} A_i} + V_{D_{x^j} A_i} + V_{[A_i, A_j]} = V_{ D_{x^j} A_i - D_{x^i}
  A_j + [A_i, A_j]}\   \Leftrightarrow{} \  D_{x^j} A_i - D_{x^i}
  A_j + [A_i, A_j]=0.
\end{multline*}
This motivates our 
choice of the minus sign in~\eqref{tilded}.
\end{remark}

\begin{example}[One\/-\/dimensional covering over the KdV equation]\label{Ex1DoverKdV}
One obtains the covering over the KdV equation from the
zero\/-\/curvature representation $\alpha_1^{\text{KdV}}$ 
(see Example~\ref{exLaxZCR} on p.~\pageref{exLaxZCR})
by using 
the realisation of Lie algebra~$\mathfrak{sl}_2(\BBC)$
in the space of vector fields.
Applying~\eqref{atensvf}, 
we construct the following vector fields with the nonlocal variable $w$:
\begin{align*}
V_A & = - (u_{12} + w^2 - \lambda)\,\tfrac{\dd}{\dd w},\\
V_B & = - \bigl( -u_{12;xx} - 2u_{12}^2 - 2\lambda u_{12} + 4\lambda^2 - 2u_{12;x}w - (2u_{12} + 4\lambda)w^2\bigr)\,\tfrac{\dd}{\dd w}.
\end{align*}
The prolongations of total derivatives act on the nonlocality~$w$ by the rules
\begin{subequations}\label{mcov}
\begin{align}
  w_x & = u_{12} + w^2 - \lambda,\label{mcov1}\\
  w_t &= -u_{12;xx} - 2u_{12}^2 - 2\lambda u_{12} + 4\lambda^2 -  2u_{12;x}w - (2u_{12} + 4\lambda)w^2.\label{mcov2}
\end{align}
\end{subequations}
This yields 
the one\/-\/dimensional covering over 
KdV equation~\eqref{kdv}.
\end{example}



\begin{example}[The projective substitution and nonlinear realisations 
of Lie algebras in the spaces of vector fields~\cite{RoelofsThesis,%
PopovychBoykoNesterenko2003,Shchepochkina2006}
]\label{exProjSub}%
Let $N$ be a $(k_0+1
)$-\/dimensional 
manifold.\footnote{\label{FootLieSuperLie}%
The realisation scheme which we outline here can be translated verbatim into 
supergeometry of su\-per\-ma\-ni\-folds~$N$ of superdimension~$(k_0+1|k_1)$ and 
zero\/-\/curvature representations with values in Lie 
su\-per\-al\-geb\-ras~$\mathfrak{g}\subseteq\mathfrak{gl}(k_0+1|k_1)$.} 
Because the reasoning is local, consider a chart~$\mathcal{U}\subseteq N$ equipped with
a $(k_0+1)$-\/tuple of rectifying coordinates 
$\boldsymbol{v} = (v^0$,\ $\dots$,\ $v^{k_0} 
)$ 
that establish a one\/-\/to\/-\/one correspondence between the points of~$\mathcal{U}$
in~$N$ and a domain in the vector space~$\BBR^{k_0+1}$. By definition, put
\[
\partial_{\boldsymbol{v}} = (\partial_{v^0}, 
\dots, \partial_{v^{k_0}} 
)^{\mathrm{t}}.  
\]
Second, let $\fg\subseteq\gl(k_0+1
)$ be a Lie algebra (see footnote~\ref{FootLieSuperLie} again). 
Take any matrix $
{g}\in\fg$
and represent it 
in the space of \emph{linear} vector fields on the domain 
in~$\BBR^{k_0+1}$ by using the formula
\[
{g}\longmapsto V_{
{g}} = \boldsymbol{v}g\partial_{\boldsymbol{v}}.
\]
By construction, the linear vector field representation $
{g}\mapsto V_{
{g}}$ of the matrix Lie
algebra~$\fg$ preserves all the 
commutation relations in it, 
\[
[
V_{
{g}}, V_{
{h}} ]
= [
\boldsymbol{v}g\partial_{\boldsymbol{v}},
\boldsymbol{v}h\partial_{\boldsymbol{v}} ]
= \boldsymbol{v}[
g,h ]
\partial_{\boldsymbol{v}} =
 V_{[
{g}, 
{h} ]
}, \qquad 
{h}, 
{g}\in\fg.
\]
The problem we are solving now is the realisation of matrix Lie algebra~$\fg$ by using vector fields with (non)\/linear coefficients.
To begin with, fix a nonzero constant $\mu\in\BBR$;
without loss of generality suppose $v^0\neq0$.
%
Consider the locally defined mapping $p\colon\BBR^{k_0+1}\to\BBR^{k_0}$ that takes every point $\boldsymbol{v}=(v^0$,\ $\ldots$,\ $v^{k_0})$
from the domain at hand 
to the point 
$(w^1,\ \dots,\ w^{k_0} 
)\in\BBR^{k_0}$, where
\[
w^i = \frac{\mu v^i}{v^0},\quad 1\leqslant i \leqslant k_0.
\]
The differential $\Id p_{\bv} \colon T_{\bv} \BBR^{k_0+1}\to T_{p(\bv)}\BBR^{k_0}$ at the point $\bv\in\BBR^{k_0+1}$ acts on the basic vectors 
from the $(k_0+1)$-\/tuple~$\dd_{\boldsymbol{v}}$ as follows,
\begin{align*}
\Id p_\bv\left( \frac{\dd}{\dd v^0}\right)
  &= \sum_{j=1}^{k_0} \frac{\dd w^j}{\dd v^0}\,\frac{\dd}{\dd w^j} 
  = \sum_{j=1}^{k_0}-\frac{\mu v^j}{(v^0)^2}\,\frac{\dd}{\dd w^j}, \\
\Id p_\bv\left( \frac{\dd}{\dd v^i}\right)
  &= \sum_{j=1}^{k_0}\frac{\dd w^j}{\dd v^i}\,\frac{\dd}{\dd w^j} = \frac{\mu}{v^0}\,\frac{\dd}{\dd w^i},
  \quad 1\leqslant i \leqslant k_0.
\end{align*}
Using these formulae,
let us calculate the action of differential $\Id p_{\bv}$ on the linear vector field $V_g$ at the point~$\bv$ in the domain:
\begin{align*}
  \Id p_\bv ( V_g) = {}&{} \Id p_\bv \left( \sum_{i,j=0}^{k_0}v^i g_{ij} \dd_{\bv}^j \right)  
  =  \sum_{i=0}^{k_0} v^i g_{i0} \sum_{j=1}^{k_0} \left(-\frac{\mu v^j}{(v^0)^2}\frac{\dd}{\dd w^j} \right) 
  + \sum_{i=0}^{k_0}\sum_{j=1}^{k_0} v^i g_{ij} \frac{\mu}{v^0}\frac{\dd}{\dd w^j} \\
{}={}&{} \sum_{i=0}^{k_0} \frac{\mu v^i}{v^0} g_{i0} \left(-\frac{1}{\mu} \sum_{j=1}^{k_0} \frac{\mu v^j}{v^0}
    \frac{\dd}{\dd w^j}\right)
  + \sum_{i=0}^{k_0}\sum_{j=1}^{k_0} \frac{\mu v^i}{v^0} g_{ij} \frac{\dd}{\dd w^j} \\
{}={}&{} \mu g_{00} \left(-\frac{1}{\mu}\sum_{j=1}^{k_0} w^j\frac{\dd}{\dd w^j} \right) 
  + \sum_{i=1}^{k_0} w^i g_{i0} \left(-\frac{1}{\mu}\sum_{j=1}^{k_0} w^j\frac{\dd}{\dd w^j} \right)  
  + \sum_{j=1}^{k_0} \mu g_{0j}\frac{\dd}{\dd w^j} \\
{}&{}
  + \sum_{i=1}^{k_0}\sum_{j=1}^{k_0} w^i g_{ij} \frac{\dd}{\dd w^j},
\end{align*}
where 
$\dd_{\bv}^j$ is the $j$th element of the tuple~$\dd_{\bv}$.
By definition, put
\[
\bw = (\mu, w^1, \ldots, w^{k_0}), \qquad 
\dd_{\bw} = \left( -\frac{1}{\mu}\sum_{j=1}^{k_0} w^j \frac{\dd}{\dd w^j}, \frac{\dd}{\dd w^1}, \ldots,
  \frac{\dd}{\dd w^{k_0}} \right)^{\mathrm{t}}.
\]
We conclude that 
the vector field $X_{
{g}}=\mathrm{d}p\,(V_{
{g}})$
is expressed by the formula 
\begin{equation}
\label{wroelofs}
X_{
{g}} = \boldsymbol{w} g \partial_{\boldsymbol{w}}.
\end{equation}
Generally speaking, the vector field~$X_{
{g}}$ on the respective subset of 
the target space $\BBR^{k_0}$ is nonlinear with respect to the variables~$w^0$,\ $\ldots$,\ %
$w^{k_0}$. 
Nevertheless, the commutation relations between vector fields of such type
are inherited 
from the relations in Lie algebra~$\fg\ni 
{g},
{h}$:
\[
[X_{
{g}}, X_{
{h}} ]
= [
\mathrm{d}p\,(V_{
{g}}),\mathrm{d}p\,(V_{
{h}}) ]
 \stackrel{(*)}{=} \mathrm{d}p\,([
V_{
{g}}, V_{
{h}} ]
) = \mathrm{d}p\,(V_{[
{g}, 
{h} ]
}) =
  X_{[
{g}, 
{h} ]
},
\]
see~\ref{appXcom} for an explicit proof of equality~$(*)$.
Take $X_
{g}$ for the representation $\varrho(
{g})$ of elements~$
{g}$ 
of Lie 
algebra~$\fg$; now that the representation~$\varrho$ is specified,
either leave the 
parameter~$\mu$ free or set it equal to any convenient nonzero constant.
We refer to~\cite{PopovychBoykoNesterenko2003,Shchepochkina2006} for other examples of realisations 
of Lie algebras by using vector fields.
\end{example}

For the sake of definition 
let us take $k_0=1$ 
so that $w^1 = w$ for 
$n = 2$ with $x^1 = x$ and~$x^2 = t$; set $\mu=1$.
Using the representation $\varrho$,  
we construct the prolongations of total derivatives,
\[
\tilde{D}_{x} =  \bar{D}_{x} + w_{x} \frac{\partial}{\partial w}, \qquad
\tilde{D}_{t} =  \bar{D}_{t} + w_{t} \frac{\partial}{\partial w},
\]
and inspect the way in which they 
act on the nonlocal variable~$w$ along $W$:
\[
w_x  = \bar{D}_x \inner \Id w, \qquad
w_t  = \bar{D}_t \inner \Id w.
\]
We thus obtain a one\/-\/dimensional covering $\tau \colon \tilde{\cE} = W\times\cEinf \to \cEinf$ with
nonlocal variable~$w$.

We claim that Gardner's deformation~\eqref{DefKdV} and zero\/-\/curvature representation~\eqref{exLaxZCR} for KdV equation~\eqref{kdv} determine the coverings which are related by using an $SL(2,\BBC)$-\/valued gauge transformation.

\begin{example}[The covering which is based on Gardner's
deformation]\label{exGardnerCov}
Consider the Gardner deformation 
of Korteweg\/--\/de Vries equation~\eqref{kdv},
\begin{align*}
\mathfrak{m}_{\veps} = & \left\{ u_{12}  = \tu_{12} - \veps \tu_{12;x} - \veps^2\tu_{12}^2
  \right\}\colon\cE_\veps\to\cE_0,\label{miura}\tag{\ref{KdVeKdV}}\\
\cE_{\veps} = & \left\{\tu_{12;t}  = - (\tu_{12;xx} + 3\tu_{12}^2 -
  2\veps^2\tu_{12}^3)_x \right\}, \label{kdvext}\tag{\ref{KdVe}}
\end{align*}
Expressing $\tu_{12;x}$ from~\eqref{miura} and substituting it
in~\eqref{kdvext}, we obtain the one\/-\/dimensional covering over the
KdV equation,
\begin{subequations}
\label{miuracov}
\begin{align}
  \tu_{12;x}& = \frac{1}{\veps}( \tu_{12} - u_{12}) - \veps\tu_{12}^2\label{miuracovx},\\
  \tu_{12;t}& = \frac{1}{\veps}( u_{12;xx} + 2u_{12}^2) + \frac{1}{\veps^2} u_{12;x} + \frac{1}{\veps^3}u_{12} 
  + \left(-2u_{12;x} - \frac{2}{\veps}u_{12} - \frac{1}{\veps^3}\right)\tu_{12} 
  + \left(2\veps u_{12} +  \frac{1}{\veps}\right)\tu_{12}^2, \label{miuracovext}
\end{align}
\end{subequations}
From this covering we derive the~$\gsl_2(\BBC)$-\/valued zero\/-\/curvature representation for~\eqref{kdv}:
\begin{equation}\label{eqKdVGDZCR}
  \alpha^{\text{KdV}}_2  = 
  \begin{pmatrix}
    \frac{1}{2\veps} & \frac{u_{12}}{\veps} \\
    -\veps & -\frac{1}{2\veps}
  \end{pmatrix}\, \Id x
  +
  \begin{pmatrix}
    u_{12;x} + \frac{1}{\veps}u_{12} + \frac{1}{2\veps^3} &
       - \frac{1}{\veps}( u_{12;xx} + 2u_{12}^2) - \frac{1}{\veps^2} u_{12;x} - \frac{1}{\veps^3}u_{12}  \\
    2\veps u_{12} +  \frac{1}{\veps} &
      - u_{12;x} - \frac{1}{\veps}u_{12} - \frac{1}{2\veps^3}
  \end{pmatrix}\, \Id t.
\end{equation}
The gauge transformation between zero\/-\/curvature representations~\eqref{zcrBVV} and~\eqref{eqKdVGDZCR} is given by the group element
\begin{equation}\label{gaugeMiura2BVV}
  S = 
  \begin{pmatrix}
    \ii/\sqrt{\veps} & \ii/(2\veps\sqrt{\veps})\\
    0 & -\ii\sqrt{\veps}
  \end{pmatrix},
\end{equation}
where we set $\lambda = \tfrac{1}{4\veps^2}$
to match the spectral parameter~$\lambda$ in~\eqref{zcrBVV} and Gardner's deformation parameter~$\varepsilon$.
\end{example}



Let us apply the same construction to Krasil'shchik\/--\/Kersten's system~\eqref{eqKK} and by this, derive the recurrence relation between the integrals of motion in its hierarchy.

\section{The 
deformation of Krasil'shchik\/--\/Kersten's system}\label{secGDKK}
From the paper~\cite{KarasuSakovichYurduen2003} we know that 
Krasil'shchik\/--\/Kersten's system~\eqref{eqKK} admits the
$\gsl_3(\BBC)$-\/valued zero\/-\/curvature representation $\alpha^{\text{KK}}_1 = A^{\text{KK}}_1\, \Id x +
B^{\text{KK}}_1 \, \Id t$, where 
\begin{gather*}
A_1^{\text{KK}} = \begin{pmatrix}
  \eta &\underline{u}{}_{12} - u_0^2 + 9\eta^2 & u_0  \\
   1 & \eta &                 0 \\
   0 & 6\eta u_0 &             -2\eta
\end{pmatrix},
\\
B_1^{\text{KK}} = \begin{pmatrix}
b_{11}  &
  b_{12} &
     - 18\eta^2u_{0}  - 3\eta u_{0;x} - u_{0;xx} + u_{0}^3 + 2u_{0} \underline{u}{}_{12}  
\\
 - 36\eta^2 + u_{0}^2 + 2\underline{u}{}_{12} &
    -b_{11} - 72\eta^3 - 6\eta u_{0}^2 & 
       - 6\eta u_{0}  - u_{0;x} 
\\
 - 36\eta^2u_{0}  + 6\eta u_{0;x} &
   b_{32} &
       72\eta^3 - 6\eta u_{0}^2
\end{pmatrix};
\end{gather*}
the elements $b_{11}$,\ $b_{12}$,\ and~$b_{32}$ of the matrix $B^{\text{KK}}_1$ are as follows:
\begin{align*}
  b_{11} = {}&{} - 36\eta^3 + 3\eta u_{0}^2 + u_{0;x}u_{0}  + \underline{u}{}_{12;x},\\
  b_{12} = {}&{}  - 324\eta^4 + 9\eta^2(u_{0}^2 - 2\underline{u}{}_{12} ) - u_{0;xx}u_{0}  - u_{0;x}^2 - \underline{u}{}_{12;xx} - u_{0}^4 -
     u_{0}^2 \underline{u}{}_{12}  + 2 \underline{u}{}_{12}^2,\\
  b_{32} = {}&{}     - 108\eta^3u_{0}  + 18\eta^2u_{0;x} + 6\eta ( - u_{0;xx} + u_{0}^3 + 2u_{0} \underline{u}{}_{12} ).
\end{align*}
Let us find 
the matrix $S^{\text{KK}} \in SL_3(\BBC) \hookrightarrow C^\infty(\cEinf, SL_3(\BBC))$ of the gauge transformation 
that makes the classical formulae by Gardner a part of the covering over system~\eqref{eqKK}.
By definition, we put $\veps = \eta^2$ in~\eqref{gaugeMiura2BVV}.
Next, let us enlarge the old group~$SL(2,\BBC)$ for equation~\eqref{kdv} to 
the gauge group~$SL(3,\BBC)$ of zero\/-\/curvature representation~$\alpha^{\text{KK}}_1$ for system~\eqref{eqKK}. We set
\[
S^{\text{KK}} = \begin{pmatrix}
  \ii \eta^{-1} & \tfrac12 \ii \eta^{-3} & 0 \\
  0             &  -\ii \eta           & 0 \\
  0             & 0                    & 1
\end{pmatrix}.
\]
Applying the gauge transformation $S^{\text{KK}}$ to the zero\/-\/curvature representation $\alpha^{\text{KK}}_1$, we obtain the gauge\/-\/equivalent 
zero\/-\/curvature representation $\alpha^{\text{KK}}_2
\mathrel{{:}{=}} (\alpha^{\text{KK}}_1)^{S^{\text{KK}}} = A^{\text{KK}}_2 \, \Id x + B^{\text{KK}}_2\,\Id t$ for Krasil'shchik\/--\/Kersten's system~\eqref{eqKK}:
\begin{gather*}
A^{\text{KK}}_2 = 
\begin{pmatrix}
  \frac23\eta^{-2} &  u_{0}^2 - \underline{u}{}_{12}    & \ii \eta^{-1}u_{0} \\
  -1               & - \frac13 \eta^{-2} & 0                  \\
   0               & \ii \eta^{-1}u_{0}  & - \frac13 \eta^{-2} 
\end{pmatrix},
\\
B^{\text{KK}}_2 = 
\begin{pmatrix}
  b_{11} &
    b_{12} &
      b_{13} \\
  - u_{0}^2 - 2 \underline{u}{}_{12}  + \eta^{-4} & 
    - u_{0;x}u_{0}  - \underline{u}{}_{12;x} - \eta^{-2} \underline{u}{}_{12}  + \frac13\eta^{-6} &
      \ii \eta^{-1}u_{0;x} + \ii \eta^{-3}u_{0} \\
   - \ii \eta^{-1}u_{0;x} + \ii \eta^{-3}u_{0} &
 \eta^{-1}( - \ii u_{0;xx} + \ii u_{0}^3 + 2\ii u_{0} \underline{u}{}_{12} ) &
        - \eta^{-2}u_{0}^2 + \frac13\eta^{-6}
\end{pmatrix},
\end{gather*}
where the elements $b_{11}$,\ $b_{12}$,\ and~$b_{13}$ of the matrix~$B^{\text{KK}}_2$ are as follows:
\begin{align*}
b_{11} = {} & {} u_{0;x}u_{0}  + \underline{u}{}_{12;x} + \eta^{-2}(u_{0}^2 
+ \underline{u}{}_{12})  - \tfrac23\eta^{-6}, 
\\
b_{12} = {} & {} u_{0;xx}u_{0} + u_{0;x}^2 + \underline{u}{}_{12;xx} + u_{0}^4 + u_{0}^2 \underline{u}{}_{12} - 2 \underline{u}{}_{12}^2 + \eta^{-2}(u_{0;x}u_{0}  
+ \underline{u}{}_{12;x}) + \eta^{-4} \underline{u}{}_{12}, \\
b_{13} = {} & {} \eta^{-1}( - \ii u_{0;xx} + \ii u_{0}^3 + 2\ii u_{0} \underline{u}{}_{12} ) - \ii \eta^{-3}u_{0;x} - \ii \eta^{-5}u_{0}.
\end{align*}
Recalling that formula~\eqref{wroelofs} yields the representation
of matrices~$A^{\text{KK}}_2$ and~$B^{\text{KK}}_2$ in terms of vector
fields, 
from the zero\/-\/curvature representation
$\alpha^{\text{KK}}_2$ we obtain the two\/-\/dimensional covering over Krasil'shchik\/--\/Kersten's system~\eqref{eqKK}.
Denoting the new nonlocal variables by~$\tu_0$ and~$\underline{\tilde{u}}{}_{12}$, we have that their derivatives with respect to the spatial variable~$x$ and time~$t$ are equal to
\begin{subequations}\label{eqKKcov}
\begin{align}
\tu_{0;x} = {} & {}  - \tu_{0} \underline{\tilde{u}}{}_{12}  - \ii \eta^{-1} u_{0}  + \eta^{-2}\tu_{0}, \label{eqKKcov0x}
\\
\underline{\tilde{u}}{}_{12;x} = {} & {}  - \underline{\tilde{u}}{}_{12}^2 - u_{0}^2 + \underline{u}{}_{12}  - \ii \eta^{-1} \tu_{0} u_{0}  
+ \eta^{-2} \underline{\tilde{u}}{}_{12}, \label{eqKKcov12x}
\\
\intertext{and}
\tu_{0;t} = {} & \tu_{0} (u_{0;x}u_{0}  + \underline{u}{}_{12;x} 
- \underline{\tilde{u}}{}_{12} u_{0}^2 - 2 \underline{\tilde{u}}{}_{12} \underline{u}{}_{12} ) 
    + \eta^{-1}(\ii u_{0;xx} - \ii u_{0;x}\tu_{0}^2 - \ii u_{0;x}\underline{\tilde{u}}{}_{12}  - \ii u_{0}^3 - 2\ii u_{0} \underline{u}{}_{12} )
    \notag \\
{}&{} + \eta^{-2}\tu_{0} (2u_{0}^2 + \underline{u}{}_{12} ) 
      + \eta^{-3}(\ii u_{0;x} + \ii \tu_{0}^2u_{0}  - \ii \underline{\tilde{u}}{}_{12} u_{0} ) 
      + \eta^{-4}\tu_{0} \underline{\tilde{u}}{}_{12} 
      + \ii \eta^{-5}u_{0}  
      - \eta^{-6}\tu_{0}, \label{eqKKcov0t}
\\
\underline{\tilde{u}}{}_{12;t} = {} & {} - u_{0;xx}u_{0}  - u_{0;x}^2 + 2u_{0;x}\underline{\tilde{u}}{}_{12} u_{0}  - \underline{u}{}_{12;xx} + 2 \underline{u}{}_{12;x} \underline{\tilde{u}}{}_{12}  -
      \underline{\tilde{u}}{}_{12}^2u_{0}^2 - 2 \underline{\tilde{u}}{}_{12}^2 \underline{u}{}_{12}  - u_{0}^4 - u_{0}^2 \underline{u}{}_{12}  \notag\\
{}&{}  + 2 \underline{u}{}_{12}^2 + \eta^{-1}\tu_{0} (\ii u_{0;xx} - \ii u_{0;x}\underline{\tilde{u}}{}_{12}  - \ii u_{0}^3 - 2\ii u_{0} \underline{u}{}_{12} ) \notag\\
{}&{}  + \eta^{-2}( - u_{0;x}u_{0}  - \underline{u}{}_{12;x} + \underline{\tilde{u}}{}_{12} u_{0}^2 + 2 \underline{\tilde{u}}{}_{12} \underline{u}{}_{12} ) \notag\\
{}&{} 
       + \ii \eta^{-3}\tu_{0} \underline{\tilde{u}}{}_{12} u_{0}  + \eta^{-4}(\underline{\tilde{u}}{}_{12}^2 - \underline{u}{}_{12} ) - \eta^{-6}\underline{\tilde{u}}{}_{12} \label{eqKKcov12t}
\end{align}
\end{subequations}
We note that under the reduction $u_0=0$ and by virtue of the relation~$u_{12}=-\underline{u}{}_{12}$, this covering retracts to Gardner's deformation~\eqref{miuracov} for KdV equation~\eqref{kdv}.

\begin{theorem}[
Gardner's deformation 
of Krasil'shchik\/--\/Kersten's system~\eqref{eqKK}]\label{thKKGD}
The extension $\cE(\veps)$ of~\eqref{eqKK} consists of, first, the evolution equation which is 
not in the form of a conserved current,
\begin{subequations}\label{eqKKGDExt}
\begin{align}
\tu_{0;t} = {} & {}  3\veps^4\tu_{0}^2\underline{\tilde{u}}{}_{12} (2\tu_{0;x}\underline{\tilde{u}}{}_{12}  + \underline{\tilde{u}}{}_{12;x}\tu_{0} ) 
+ 3\veps^3\tu_{0} ( - \tu_{0;xx}\tu_{0} \underline{\tilde{u}}{}_{12}  - 3\tu_{0;x}^2\underline{\tilde{u}}{}_{12}  - \tu_{0;x}\underline{\tilde{u}}{}_{12;x}\tu_{0} )\notag
\\
{}&{}  + 3\veps^2(\tu_{0;xx}\tu_{0;x}\tu_{0}  + \tu_{0;x}^3 + 3\tu_{0;x}\tu_{0}^2\underline{\tilde{u}}{}_{12}  + \tu_{0;x}\underline{\tilde{u}}{}_{12}^2 + \underline{\tilde{u}}{}_{12;x}\tu_{0}^3 + \underline{\tilde{u}}{}_{12;x}\tu_{0} \underline{\tilde{u}}{}_{12} ) \notag
\\ 
{}&{} + 3\veps \tu_{0} ( - \tu_{0;xx}\tu_{0}  - 2\tu_{0;x}^2) - \tu_{0;xxx} + 3\tu_{0;x}\tu_{0}^2 + 3\tu_{0;x}\underline{\tilde{u}}{}_{12}  + 3\underline{\tilde{u}}{}_{12;x}\tu_{0},
\end{align}
and second, the continuity relation
\begin{align}
\underline{\tilde{u}}{}_{12;t} = {} & {} 
\frac{\Id}{\Id x} \Bigl( 
    3\veps^4\tu_{0}^2\underline{\tilde{u}}{}_{12}^3 
  + 3\veps^3\tu_{0} \underline{\tilde{u}}{}_{12} (\tu_{0;x}\underline{\tilde{u}}{}_{12}  - \underline{\tilde{u}}{}_{12;x}\tu_{0} ) 
  + \veps^2( - 3\tu_{0;xx}\tu_{0} \underline{\tilde{u}}{}_{12} - 3\tu_{0;x}\underline{\tilde{u}}{}_{12;x}\tu_{0}  
     + 2\underline{\tilde{u}}{}_{12}^3  \notag \\
{}&{} 
     + 6\tu_{0}^2\underline{\tilde{u}}{}_{12}^2) 
  + 3\veps ( \tu_{0;xx}\tu_{0;x} - \tu_{0;x}\tu_{0} \underline{\tilde{u}}{}_{12}  + \underline{\tilde{u}}{}_{12;x}\tu_{0}^2) 
  - 3\tu_{0;xx}\tu_{0}  - \underline{\tilde{u}}{}_{12;xx} + 3\tu_{0}^2\underline{\tilde{u}}{}_{12}  + 3\underline{\tilde{u}}{}_{12}^2 \Bigr).
\end{align}
\end{subequations}
The Miura contraction from~\eqref{eqKKGDExt} to~\eqref{eqKK} is
\begin{subequations}\label{eqKKGDMiura}
\begin{align}
u_{0} = {} & {}  \tu_{0} - \veps \tu_{0;x} + \veps^2\underline{\tilde{u}}{}_{12} \tu_{0},
\\
\underline{u}{}_{12} = {} & {}  \underline{\tilde{u}}{}_{12} 
      - \veps(\underline{\tilde{u}}{}_{12;x} + \tu_{0;x}\tu_{0} )
      + \veps^2(\tu_{0;x}^2 + \underline{\tilde{u}}{}_{12}^2 + \underline{\tilde{u}}{}_{12} \tu_{0}^2) 
      - 2\veps^3u_{0;x}\underline{\tilde{u}}{}_{12} \tu_{0}  
      + \veps^4\underline{\tilde{u}}{}_{12}^2\tu_{0}^2. 
\end{align}
\end{subequations}
Under the reduction~$u_0=0$, this deformation retracts to classical Gardner's formulas~\eqref{DefKdV}. 
\end{theorem}

\begin{proof}
Let us express~$u_0$ and $\underline{u}{}_{12}$ from~(\ref{eqKKcov0x}-\ref{eqKKcov12x}) and plug them in~(\ref{eqKKcov0t}-\ref{eqKKcov12t}). We get
\begin{align*}
u_{0} = {} & {} \ii \eta (\tu_{0;x} + \tu_{0} \underline{\tilde{u}}{}_{12} ) - \ii \eta^{-1}\tu_{0} ,
\\
\underline{u}{}_{12} =  {} & {} \eta^2( - \tu_{0;x}^2 - 2\tu_{0;x}\tu_{0} \underline{\tilde{u}}{}_{12}  - \tu_{0}^2\underline{\tilde{u}}{}_{12}^2) +
\tu_{0;x}\tu_{0}  + \underline{\tilde{u}}{}_{12;x} + \tu_{0}^2\underline{\tilde{u}}{}_{12}  + \underline{\tilde{u}}{}_{12}^2 - \eta^{-2}\underline{\tilde{u}}{}_{12},
\\
u_{0;t} = {} & {} 3\eta^2( - \tu_{0;xx}\tu_{0;x}\tu_{0}  - \tu_{0;xx}\tu_{0}^2\underline{\tilde{u}}{}_{12}  - \tu_{0;x}^3 -
       3\tu_{0;x}^2\tu_{0} \underline{\tilde{u}}{}_{12}  - \tu_{0;x}\underline{\tilde{u}}{}_{12;x}\tu_{0}^2 - 2\tu_{0;x}\tu_{0}^2\underline{\tilde{u}}{}_{12}^2 \\ 
    {} & {} - \underline{\tilde{u}}{}_{12;x}\tu_{0}^3\underline{\tilde{u}}{}_{12} )- \tu_{0;xxx} + 3\tu_{0;xx}\tu_{0}^2 + 6\tu_{0;x}^2\tu_{0}  + 9\tu_{0;x}\tu_{0}^2\underline{\tilde{u}}{}_{12}  +
    3\tu_{0;x}\underline{\tilde{u}}{}_{12}^2 + 3\underline{\tilde{u}}{}_{12;x}\tu_{0}^3  \\
    {} &{} + 3\underline{\tilde{u}}{}_{12;x}\tu_{0} \underline{\tilde{u}}{}_{12}  - 3\eta^{-2}(\tu_{0;x}\tu_{0}^2 + \tu_{0;x}\underline{\tilde{u}}{}_{12}  + \underline{\tilde{u}}{}_{12;x}\tu_{0} ),
\\
\underline{u}{}_{12;t} = {} & {} 3\eta^2(\tu_{0;xxx}\tu_{0;x} + \tu_{0;xxx}\tu_{0} \underline{\tilde{u}}{}_{12}  + \tu_{0;xx}^2 +
    \tu_{0;xx}\tu_{0;x}\underline{\tilde{u}}{}_{12}  + 2\tu_{0;xx}\underline{\tilde{u}}{}_{12;x}\tu_{0} - \tu_{0;xx}\tu_{0} \underline{\tilde{u}}{}_{12}^2 \\
    {} & {} + \tu_{0;x}^2\underline{\tilde{u}}{}_{12;x} - \tu_{0;x}^2\underline{\tilde{u}}{}_{12}^2 + \tu_{0;x}\underline{\tilde{u}}{}_{12;xx}\tu_{0}  - 2\tu_{0;x}\tu_{0}
    \underline{\tilde{u}}{}_{12}^3 + \underline{\tilde{u}}{}_{12;xx}\tu_{0}^2\underline{\tilde{u}}{}_{12}  + \underline{\tilde{u}}{}_{12;x}^2\tu_{0}^2 - 3\underline{\tilde{u}}{}_{12;x}\tu_{0}^2\underline{\tilde{u}}{}_{12}^2) \\
    {} & {}- 3\tu_{0;xxx}\tu_{0}  - 3\tu_{0;xx}\tu_{0;x} + 3\tu_{0;xx}\tu_{0} \underline{\tilde{u}}{}_{12}  + 3\tu_{0;x}^2\underline{\tilde{u}}{}_{12}
    - 3\tu_{0;x}\underline{\tilde{u}}{}_{12;x}\tu_{0}  + 12\tu_{0;x}\tu_{0} \underline{\tilde{u}}{}_{12}^2 \\ 
    {} & {} - \underline{\tilde{u}}{}_{12;xxx} - 3\underline{\tilde{u}}{}_{12;xx}\tu_{0}^2 + 12\underline{\tilde{u}}{}_{12;x}\tu_{0}^2\underline{\tilde{u}}{}_{12}  + 6\underline{\tilde{u}}{}_{12;x}\underline{\tilde{u}}{}_{12}^2 \\
    {} & {}+ 3\eta^{-2}( - 2\tu_{0;x}\tu_{0} \underline{\tilde{u}}{}_{12}  - \underline{\tilde{u}}{}_{12;x}\tu_{0}^2 - 2\underline{\tilde{u}}{}_{12;x}\underline{\tilde{u}}{}_{12} ).
\end{align*}
Setting $\tu_{0;\text{new}} = \ii \eta\,\tu_{0;\text{old}}$ 
and $\underline{\tilde{u}}{}_{12;\text{new}} = \eta^2\underline{\tilde{u}}{}_{12;\text{old}}$, 
and putting $\eta=\sqrt{\veps}$, 
we derive formulas~(\ref{eqKKGDExt}-\ref{eqKKGDMiura}).
\end{proof}


\begin{remark}
This change of variables removes the singularity at $\eta=\veps=0$ in
formulas~\eqref{eqKKcov}. (Let us recall that Gardner's deformations with a singularity at
$\veps=0$ do yield the recurrence relations between the conserved
densities, see~\cite{JMP2012,AOKThesis} for an example of such deformation
for the $N{=}2$ supersymmetric $a{=}4$-Korteweg\/--\/de Vries equation.)
\end{remark}

In Theorem~\ref{thKKGD} we obtained the deformation of Krasil'shchik\/--\/Kersten's system such that one of the extended equations is not a continuity relation, which makes the construction different from the classical concept of Gardner's deformation. 
Nevertheless, deformations of such unconventional type do yield the recurrence relations between the conserved densities.

\begin{theorem}
Gardner's deformation~(\ref{eqKKGDExt}-\ref{eqKKGDMiura}) 
for Krasil'shchik\/--\/Kersten's system~\eqref{eqKK} 
yields the following recurrence relations between the conserved densities, 
which we denote by~$\underline{w}{}_n$ for~$n\in\BBN_{\geqslant0}$\textup{:}
\begin{align*}
\underline{w}{}_0 = {} & {} \underline{u}{}_{12}, \qquad 
\underline{w}{}_1 = \underline{u}{}_{12;x} + u_{0;x}u_0,\\
\underline{w}{}_2 = {} & {} D_x \underline{w}{}_1  + D_x (v_0v_1) - u_{0;x}^2 - \underline{u}{}_{12}^2 - \underline{u}{}_{12} u_{0}^2,\\
\underline{w}{}_3 = {} & {} D_x \underline{w}{}_2 + \frac12\sum_{k=0}^2 D_x (v_k v_{2-k}) - 2 D_x( v_0)D_x( v_1) - 2 \underline{w}{}_1 \underline{w}{}_0 - \underline{w}{}_1 v_0^2 
        - 2\underline{w}{}_0 v_1v_0 + 2u_{0;x}u_{0} \underline{u}{}_{12},
\\
\underline{w}{}_n = {} & {} D_x \underline{w}{}_{n-1} + \frac12\sum_{k=0}^{n-1} D_x(v_k v_{n-1-k}) 
        - \sum_{k=0}^{n-2} \left( D_x(v_k) D_x(v_{n-2-k}) - \underline{w}{}_k \underline{w}{}_{n-2-k} \right) \\
      {} & {} - \sum_{k+l+j=n-2} \underline{w}{}_k v_l v_j  + 2\sum_{k+l+j=n-3} \underline{w}{}_k v_l D_x  v_j
        - \sum_{k+l+j+i=n-4} \underline{w}{}_k \underline{w}{}_l v_j v_i,\qquad n\geqslant4,
\end{align*}
where the quantities~$v_i$ are given by the formulas
\[
v_0 = {}  u_0, \qquad v_1 = u_{0;x}, \qquad v_n = D_x v_{n-1} - \sum_{k=0}^{n-2} \underline{w}{}_k v_{n-2-k}\quad \text{for $n\geqslant2$.}
\]
The generating function $\underline{\breve{w}}(u_0,\underline{u}{}_{12},\veps)$ of the zero differential order component of the series $w([u_0,\underline{u}{}_{12}],\veps)$
is given by the formula
\begin{equation}\label{eqKKGDGenFun}
\underline{\breve{w}} = \frac{12\veps^2( - u_0^2 + \underline{u}{}_{12}) + q^2 - 4q + 4}{6\veps^2 q},
\end{equation}    
where we put
\begin{multline*}
q = 2^{2/3}\Bigl(9\veps^2(2u_{0}^2 + \underline{u}{}_{12})  + 2\\
    + 3\sqrt{3}\veps\sqrt{4\veps^4(u_{0}^6 - 3u_{0}^4 \underline{u}{}_{12} + 3u_{0}^2 \underline{u}{}_{12}^2 - \underline{u}{}_{12}^3) + \veps^2(8u_{0}^4 + 20u_{0}^2 \underline{u}{}_{12} - \underline{u}{}_{12}^2) + 4u_{0}^2} \Bigr)^{1/3}.
\end{multline*}
\end{theorem}

\begin{proof}
Plugging the series $\tu_0 = \sum_{k=0}^{+\infty} \veps^k v_k$ 
and $\underline{\tilde{u}}{}_{12} = \sum_{k=0}^{+\infty} \veps^k \underline{w}{}_k$ 
into~\eqref{eqKKGDExt}, 
we obtain the recurrence relations between~$v_k$ and~$\underline{w}{}_k$. The 
coefficients~$\underline{w}{}_k$ are conserved because $\underline{\tilde{u}}{}_{12;t}$~is in divergent form (i.e., the velocity of $\underline{\tu}{}_{12}$ is in the image of $\Id/\Id x$).
The coefficients~$v_k$ are auxiliary quantities which, generally speaking, could 
be not conserved 
(note that the density $v_0$ is also conserved by virtue of~\eqref{eqKK} at $\veps=0$).

The respective zero\/-\/order components~$\breve{v}(u_0,\underline{u}{}_{12},\veps)$ 
and~$\underline{\breve{w}}(u_0,\underline{u}{}_{12},\veps)$ of~$\tu_0$ and~$\underline{\tu}{}_{12}$
satisfy the equations 
\begin{subequations}
\begin{align}
u_0 = {} & {} \breve{v} + \veps^2 \breve{v}\underline{\breve{w}}, \label{eqKKGenFunEq1}\\
\underline{u}{}_{12} = {} & {} \underline{\breve{w}} + \veps^2(\breve{v}^2 \underline{\breve{w}} + \underline{\breve{w}}{}^2) + \veps^4\underline{\breve{w}}{}^2\breve{v}^2. \label{eqKKGenFunEq2}
\end{align}
\end{subequations}
Expressing~$\breve{v}$ from~\eqref{eqKKGenFunEq1}, we obtain that
\[
\breve{v} = \frac{u_0}{1 + \veps^2 \underline{\breve{w}}}.
\]
Substituting this expression 
further in~\eqref{eqKKGenFunEq2}, we obtain the cubic 
equation for~$\underline{\breve{w}}$,
\[
\veps^4\underline{\breve{w}}{}^3+2\veps^2\underline{\breve{w}}{}^2+(\veps^2u_0^2-\veps^2 \underline{u}{}_{12}+1)\underline{\breve{w}} - \underline{u}{}_{12} = 0.
\]
The limit behaviour $\lim_{\veps\to0} \underline{\breve{w}} = \underline{u}{}_{12}$
of its solution at the origin of deformation parameter prescribes that
we 
pick the root~\eqref{eqKKGDGenFun}. 
\end{proof}

Let us calculate several low\/-\/order conserved densities for system~\eqref{eqKK} by using the recurrence relations for~$\underline{w}_n$:
\begin{align*}
\underline{w}{}_0 = {} & \underline{u}{}_{12},\\
\underline{w}{}_1 = {} & u_{0;x}u_0 + \underline{u}{}_{12;x}, \\
\underline{w}{}_2 = {} & - \underline{u}{}_{12}u_{0}^2 - \underline{u}{}_{12}^2+2u_0u_{0;xx}+u_{0;x}^2+u_{0;xx}, 
\\
\underline{w}{}_3 = {} & -u_{0}^3u_{0;x}-6 \underline{u}{}_{12}u_{0}u_{0;x} - 3\underline{u}{}_{12;x}u_{0}^2 - 4\underline{u}{}_{12}\underline{u}{}_{12;x} + 3u_{0}u_{0;xxx} +  4u_{0;x}u_{0;xx} + \underline{u}{}_{12;xxx}, 
\\
\underline{w}{}_4 = {} & \underline{u}{}_{12}u_{0}^4 + 4\underline{u}{}_{12}^2u_{0}^2 - 4u_{0}^3u_{0;xx} - 8u_{0}^2u_{0;x}^2 
    + 2\underline{u}{}_{12}^3 - 13\underline{u}{}_{12}u_{0}u_{0;xx} - 8\underline{u}{}_{12}u_{0;x}^2 - 19\underline{u}{}_{12;x}u_{0}u_{0;x} 
\\ 
  {} &{} - 6\underline{u}{}_{12;xx}u_{0}^2  - 6\underline{u}{}_{12}\underline{u}{}_{12;xx} - 5\underline{u}{}_{12;x}^2 + 4u_{0}u_{0;4x} + 7u_{0;x}u_{0;xxx} + 4u_{0;xx}^2 + \underline{u}{}_{12;4x},\qquad \text{etc.}
\end{align*}

\begin{lemma}
The conserved densities $\underline{w}{}_{2k}$ with even indexes $2k\geqslant0 $ are non\/-\/trivial.
\end{lemma}

This follows from the non\/-\/triviality of conserved densities $\tu_{12}^{(2k)}$ generated at all even indexes
by Gardner's deformation~\eqref{DefKdV}
for KdV equation~\eqref{kdv}.

\begin{proof}
The reduction $u_{0}=0$ maps Krasil'shchik\/--\/Kersten's system~\eqref{eqKK} to the Korteweg\/--\/de Vries equation; likewise, the reduction $\tilde{u}_0=0$ takes the recurrence relations between the conserved 
densities~$\underline{w}{}_k$ 
for~\eqref{eqKK} to Gardner's formulas for the 
densities $\tu_{12}^{(k)}$ which are conserved on equation~\eqref{kdv}.
Therefore, the quantities~$\underline{w}{}_{2k}$ have the form
\[
\underline{w}{}_{2k}\bigl([\underline{u}{}_{12}],[u_{0}]\bigr) 
 =
\tu_{12}^{(2k)}\bigl([\underline{u}{}_{12}]\bigr) +
f_{2k}\bigl([\underline{u}{}_{12}],[u_0]\bigr), \qquad k\in\BBN,
\]
where the differential polynomials 
$f_{2k}\bigl([\underline{u}{}_{12}],[u_0]\bigr)$ are 
such that $\left.f_{2k}\right|_{u_0=0}=0$.
On the one hand, the densities~$\tu_{12}^{(2k)}$ are known to be equal to
$\tu_{12}^{(2k)} = c_k \cdot
\underline{u}{}_{12}^k + \ldots$, where
$c_k$ are nonzero constants~\cite{KuperOnNature}.
On the other hand, every monomial in $f_{2k}$
essentially depends on either~$u_0$ or its derivatives with respect to $x$,
hence $f_{2k}$ cannot contain the monomial~$\underline{u}{}_{12}^k$.
Consequently, 
$\underline{w}{}_{2k} = c_k \cdot \underline{u}{}_{12}^k + \ldots$ with~$c_k\neq0$, so that these quantities may not belong to the image of total derivative $\Id/\Id x$, whence they
are nontrivial.
\end{proof}

It must be expected that the densities $\underline{w}{}_{2k+1}$ with odd indexes, not contributing to the hierarchy of~\eqref{eqKK}, are trivial; this is confirmed by a straightforward calculation of small\/-\/index terms in that auxiliary sequence. The triviality of irrelevant quantities $\underline{w}{}_{2k+1}$, $k\in\BBN$ can be approached, e.g., by using the technique from~\cite{KuperOnNature}; that method's idea is a realisation of the generating function for --\,in retrospect, trivial\,-- conserved densities $\underline{w}{}_{2k+1}$ via the non\/-\/trivial quantities $\underline{w}{}_{2k}$ and the generating function for them. 

In the meantime, the problem of finding recurrence relations for the hierarchy of integrals of motion for the Krasil'shchik\/--\/Kersten system is solved.

\begin{remark}
In the recent papers~\cite{TMPh2006,JPCS14} 
it was shown that Gardner's deformations provide the 
initial data for construction of new integrable systems. 
Applying the algorithm described in~\cite{TMPh2006,JPCS14}
to 
Gardner's deformation~(\ref{eqKKGDExt}-\ref{eqKKGDMiura}) 
for Krasil'shchik\/--\/Kersten's system, we obtain
the Kaup\/--\/Newell hierarchy~\cite{KaupNewell}.
\end{remark}

Let us conclude this paper by recalling that the recurrence relation for the hierarchy of Hamiltonian functionals for the bosonic\/-\/limit system~\eqref{eqKK} is the initial datum for solution of Gardner's deformation problem for the full $N{=}2$ supersymmetric $a{=}1$ Korteweg\/--\/de Vries equation from~\cite{MathieuNew,MathieuOpen}. This technique of recursive construction of the Hamiltonian super\/-\/functionals that depend on the $N{=}2$ superfield was developed in~\cite{HKKW}.

\appendix
\section{}\label{appXcom}
Let $\fg\subseteq\gl(k_0+1)$ be a finite-dimensional 
Lie algebra and $
{g}\in\fg$. 
On the one hand, the element~$g$ is represented
in the space of $(k_0+1)\times(k_0+1)$ matrices. 
On the other hand, the element~$
{g}$ can be represented
in the space of vector fields on some open domain in $\BBR^{k_0+1}$ by using the formula
\[
{g}\longmapsto X_{
{g}} = \bw g \pb,
\]
where
\[
\bw = \begin{pmatrix} 
  \mu & w^1 & \dots & w^{k_0}
\end{pmatrix},
\qquad
\pb = \begin{pmatrix}
  -\frac1\mu\sum_{i=1}^{k_0} w^i\dd_{w^i} \\
  \dd_{w^1} \\
  \vdots \\
  \dd_{w^{k_0}}
\end{pmatrix}.
\]
Let us show that the vector field representation preserves all the commutation relations in~$\fg$, that is, let us verify the identity
\[
[X_{
{g}}, X_{
{h}}] = X_{
{g}}( X_{
{h}}) - X_{
{h}} (X_{
{g}}) = X_{[
{g}, 
{h}]}
\]
for all $
{g}, 
{h}\in\fg$.

By convention, summation over repeated indexes is performed (here, not necessarily over one \emph{upper} and one \emph{lower} index, which is due to the notation for matrix elements of~$
{g}$,\ $
{h}$ and for the components of~$\dd_{\bw}$).

\begin{proof}
We have that
\begin{align*}
X_{
{g}}(X_{
{h}}) = {}& \bw^i g_{ij}  \pb^j  \left( \bw^p h_{pq} \pb^q \right) 
\\
{}={}&
\bw^i g_{i0} \pb^0 (\bw^p h_{pq} \pb^q) + \sum_{j\neq0 } \bw^i g_{ij} \pb^j \left( \bw^p h_{pq} \pb^q \right) 
\\
{}={}&\sum_{p\neq 0}\bw^i g_{i0} \left(-\tfrac1\mu \bw^p\right) h_{pq} \pb^q + 
\bw^i g_{i0} \bw^p h_{p0} \left( -\tfrac1\mu \pb^0\right)
\\
{}&{}\qquad {}  
+ \sum_{j\neq0} \bw^i g_{ij} h_{jq} \pb^q
+ \sum_{j\neq0} \bw^i g_{ij} \bw^p h_{p0} \left( -\tfrac1\mu \pb^j \right)
\\
{}={}& \sum_{p\neq 0}\bw^i g_{i0} \left(-\tfrac1\mu \bw^p\right) h_{pq} \pb^q 
-  \bw^p h_{p0} g_{00} \pb^0 + \sum_{i\neq0} \bw^i g_{i0} \bw^p h_{p0} \left( -\tfrac1\mu \pb^0\right)  
+ \sum_{j\neq0} \bw^i g_{ij} h_{jq} \pb^q 
\\
{}&{} \qquad{}
- \sum_{j\neq0}\bw^p h_{p0} g_{0j} \pb^j
+ \sum_{j\neq0,i\neq 0} \bw^i g_{ij} \bw^p h_{p0} \left( -\tfrac1\mu \pb^j \right) 
\\
{}={}& \sum_{j\neq0} \bw^i g_{ij} h_{jq} \pb^q  - \bw^p h_{p0}g_{0j} \pb^j
 + \sum_{i\neq 0} \bw^i g_{ij} \bw^p h_{p0} \left( -\tfrac1\mu \pb^j \right) 
 + \sum_{p\neq 0}\bw^i g_{i0} \left(-\tfrac1\mu \bw^p\right) h_{pq} \pb^q.
\end{align*}
An almost the same calculation yields that
\begin{align*}
X_{
{h}}(X_{
{g}}) = {}& \bw^p h_{pq}  \pb^q  \left( \bw^i g_{ij} \pb^j \right) 
\\
{}= {}&
 \sum_{q\neq0} \bw^p h_{pq} g_{qj} \pb^j  - \bw^i g_{i0}h_{0q} \pb^q
 + \sum_{p\neq 0} \bw^p h_{pq} \bw^i g_{i0} \left( -\tfrac1\mu \pb^q \right) 
 + \sum_{i\neq 0}\bw^p h_{p0} \left(-\tfrac1\mu \bw^i\right) g_{ij} \pb^j
\end{align*}
Therefore, the commutator of vector fields $X_{
{g}}$ and $X_{
{h}}$ is equal to
\begin{align*}
[X_{
{g}}, X_{
{h}}] = {}& 
\bw^i g_{ij}  \pb^j  \left( \bw^p h_{pq} \pb^q \right) -  \bw^p h_{pq}  \pb^q  \left( \bw^i  f_{ij} \pb^j\right)
\\
{}={}& \sum_{j\neq0} \bw^i g_{ij} h_{jp} \pb^p  - \bw^p h_{p0}g_{0j} \pb^j 
 - \sum_{q\neq0} \bw^p h_{pq} g_{qj} \pb^j  + \bw^i g_{i0}h_{0q} \pb^q
\\
{}={}& \bw^i g_{ij} h_{jp} \pb^p - \bw^p h_{pq} g_{qj} \pb^j   
 = \bw^i \left( g_{ij}h_{jq} - h_{ip}g_{pq} \right) \pb^q
 = \bw^i \left([g,h]\right)_{iq} \pb^q = X_{[
{g},
{h}]}.
\end{align*}
This proves the claim.
\end{proof}

\ack
The authors thank the Organising committee of the 7th 
International workshop
`Group Analysis of Differential Equations and Integrable Systems' 
(15--19 June 2014, Larnaca, Cyprus) for 
a warm atmosphere during the meeting.
The authors thank M.~A.~Nesterenko, D.~A.~Leites, and A.~V.~Mi\-khai\-lov 
for stimulating discussions and constructive criticism;
the authors are grateful to the anonymous referee for 
remarks 
and advice.
The research of A.\,V.\,K.\ was partially supported by JBI~RUG project~103511 (Groningen); the research of
A.\,O.\,K.\ was partially supported by JBI~RUG project~135110 (Groningen).


\section*{References}

\medskip
\providecommand{\newblock}{}

\end{document}